\documentclass[10pt]{article}
\usepackage{palatino}
\usepackage{float}
\usepackage[english]{babel}

\usepackage{url}
\usepackage[colorlinks = true]{hyperref}
\hypersetup{citecolor = black}
\usepackage{fullpage,amsmath,amsthm,amssymb,amsbsy}
\usepackage{mathdots}
\usepackage{paralist}
\usepackage{xcolor}
\usepackage{color}
\usepackage{graphicx}
\usepackage{algorithm,algpseudocode}
\usepackage{comment}

\usepackage{fancyhdr}
\usepackage{cite}
\usepackage{cleveref}

\newtheorem{thm}{Theorem}
\newtheorem{definition}{Definition}

\newtheorem{lem}{Lemma}
\theoremstyle{remark}
\newtheorem{remark}{Remark}

\newcommand{\R}{\mathbb{R}}
\def \real    { \mathbb{R} }

\newcommand{\C}{\mathbb{C}}

\newcommand{\Z}{\mathbb{Z}}
\newcommand{\N}{\mathbb{N}}

\newcommand{\vct}[1]{\boldsymbol{#1}}
\newcommand{\mtx}[1]{\boldsymbol{#1}}

\renewcommand{\H}{\mathrm{H}}
\newcommand{\T}{\mathrm{T}}

\newcommand{\Span}{\operatorname{Span}}
\newcommand{\trace}{\operatorname{trace}}

\newcommand{\rank}{\operatorname{rank}}

\DeclareMathOperator*{\minimize}{\text{minimize}}

\newcommand{\eps}{\epsilon}

\newcommand{\calA}{\mathcal{A}}
\newcommand{\calB}{\mathcal{B}}

\newcommand{\calS}{\mathcal{S}}
\newcommand{\calT}{\mathcal{T}}

\newcommand{\va}{\vct{a}}
\newcommand{\vb}{\vct{b}}

\newcommand{\ve}{\vct{e}}

\newcommand{\vs}{\vct{s}}

\newcommand{\vx}{\vct{x}}
\newcommand{\vy}{\vct{y}}

\newcommand{\valpha}{\vct{\alpha}}

\newcommand{\mA}{\mtx{A}}
\newcommand{\mB}{\mtx{B}}

\newcommand{\mD}{\mtx{D}}
\newcommand{\mE}{\mtx{E}}
\newcommand{\mF}{\mtx{F}}

\newcommand{\mL}{\mtx{L}}

\newcommand{\mP}{\mtx{P}}
\newcommand{\mQ}{\mtx{Q}}

\newcommand{\mS}{\mtx{S}}
\newcommand{\mT}{\mtx{T}}
\newcommand{\mU}{\mtx{U}}
\newcommand{\mV}{\mtx{V}}

\newcommand{\mY}{\mtx{Y}}

\newcommand{\mLambda}{\mtx{\Lambda}}
\newcommand{\mOmega}{\mtx{\Omega}}
\newcommand{\mPhi}{\mtx{\Phi}}

\newcommand{\mId}{{\bf I}}

\newcommand{\mzero}{{\bf 0}}

\graphicspath{{./figs/}}

\newlength{\imgwidth}
\newboolean{twoColVersion}
\setboolean{twoColVersion}{false}
\newcommand{\twoCol}[2]{\ifthenelse{\boolean{twoColVersion}} {#1} {#2} }

\newcommand\blfootnote[1]{
  \begingroup
  \renewcommand\thefootnote{}\footnote{#1}
  \addtocounter{footnote}{-1}
  \endgroup
}

\begin{document}

\title{ROAST: Rapid Orthogonal Approximate Slepian Transform}

\vspace{2mm}
\author{Zhihui Zhu, Santhosh Karnik, Michael B. Wakin, Mark A. Davenport, Justin  Romberg}

\maketitle

\begin{abstract}
In this paper, we provide a Rapid Orthogonal Approximate Slepian Transform (ROAST) for the discrete vector that one obtains when collecting a finite set of uniform samples from a baseband analog signal. The ROAST offers an orthogonal projection which is an approximation to the orthogonal projection onto the leading discrete prolate spheroidal sequence (DPSS) vectors (also known as Slepian basis vectors). As such, the ROAST is guaranteed to accurately and compactly represent not only oversampled bandlimited signals but also the leading DPSS vectors themselves. Moreover, the subspace angle between the ROAST subspace and the corresponding DPSS subspace can be made arbitrarily small. The complexity of computing the representation of a signal using the ROAST is comparable to the FFT, which is much less than the complexity of using the DPSS basis vectors. We also give non-asymptotic results to guarantee that the proposed basis not only provides a very high degree of approximation accuracy in a mean squared error sense for bandlimited sample vectors, but also that it can provide high-quality approximations of all sampled sinusoids within the band of interest.
\end{abstract}

\blfootnote{ Z. Zhu is with the Center for Imaging Science at the Johns Hopkins University, Baltimore, MD, USA (email: zzhu29@jhu.edu). S. Karnik, M. A. Davenport, and J. Romberg are with the School of Electrical and Computer Engineering, Georgia Institute of Technology, Atlanta, GA, USA (e-mail: skarnik1337@gatech.edu, mdav@gatech.edu, jrom@ece.gatech.edu). M. Wakin is with the  Department of Electrical Engineering, Colorado School of Mines, Golden, CO, USA (e-mail: mwakin@mines.edu).
This work was supported by NSF grants CCF-1409261 and CCF-1409406.  A preliminary version of this paper highlighting some of the key results also appeared in~\cite{zhu2017fast}.}
\section{Introduction}
\label{sec:intro}
The Nyquist-Shannon sampling theorem guarantees that real world signals that are bandlimited (or can be made bandlimited by filtering) can be replaced by a discrete sequence of their samples without the loss of any information. These samples can then be processed digitally. In particular, the discrete Fourier transform (DFT) for digital signals has been widely used for many applications in engineering, mathematics, and science thanks to the fast Fourier tranform (FFT), an efficient algorithm for computing the DFT.

Due to the fact that finite windowing in the time domain will spread out a signal's spectrum in the frequency domain, however, the DFT suffers from frequency leakage when used to represent a finite-length vector arising from a bandlimited signal with a narrowband spectrum, or even a pure sinusoid. This problem can be mitigated to some degree by applying a smooth windowing function in the sampling system. Alternatively, one can compactly represent the signals using a basis of timelimited discrete prolate spheroidal sequences (DPSS's). DPSS's, first introduced by Slepian in 1978~\cite{Slepian78DPSS}, are a collection of orthogonal bandlimited sequences that are most concentrated in time to a given index range. When limited in the time domain, they provide a compact (and again orthogonal) representation for sampled bandlimited signals.

Owing to their concentration in the time and frequency domains, the DPSS's have been successfully used in numerous signal processing applications. For instance,  DPSS's can be applied to find the minimum energy, infinite-length bandlimited sequence that extrapolates a given finite timelimited vector of samples~\cite{Slepian78DPSS}; bandlimited extrapolation is a classical signal processing problem and appears in applications such as spectral estimation and image processing~\cite{papoulis1975new,hayes1983bandlimited}. Another problem involves estimating time-varying channels in wireless communication systems. In~\cite{zemen2005channelEstim,zemen2007minimum}, Zemen et al.\ showed that expressing the time-varying subcarrier coefficients with a DPSS basis yields better estimates than those obtained with a DFT basis, which suffers from frequency leakage.  In through-the-wall radar imaging using stepped-frequency synthetic aperture radar (SAR)~\cite{amin2008wideband}, the DPSS basis can be utilized for efficiently mitigating wall clutter and for detecting targets behind the wall~\cite{AhmadQianAmin2015WallCluterDPSS,Zhu2015targetDetectDPSS,Zhu2016targetDetectDPSS}. In addition, DPSS's are useful for multiband signal identification~\cite{Zhu17Asilomar} and narrowband and multiband signal recovery from compressive measurements~\cite{DavenSSBWB_Wideband,DavenportWakin2012CSDPSS}. Building on this, DPSS's have been used to enable compressive sensing of physiological signals~\cite{sejdic2012compressive}. More broadly, the ability to recover multiband signals is beneficial for developing high-bandwidth radio receivers for cognitive radio and communications intelligence~\cite{wakin2012nonuniform}.

Unfortunately, unlike the DFT which can be computed efficiently with the FFT algorithm, there exists no  algorithm that can efficiently compute the DPSS representation for a very large signal. Recently, we proposed~\cite{Karnik2016FAST} a fast Slepian transform (FST), a fast method for computing approximate projections onto the leading DPSS vectors and compressing a signal to the corresponding low dimension. Despite its favorable properties, the fast algorithm presented in~\cite{Karnik2016FAST} did not correspond to an orthogonal projection. In this paper, we illustrate an alternative {\em orthonormal} basis that provides an approximate but sufficiently accurate representation of the subspace spanned by the leading DPSS vectors and compactly captures most of the energy in oversampled bandlimited signals. The representation of an arbitrary vector in this basis can again be computed efficiently (with complexity comparable to that of the FFT), and we refer to this procedure as the Rapid Orthogonal Approximate Slepian Transform (ROAST).

One of the main contributions of this paper is to confirm that such an orthonormal basis not only provides a very high degree of approximation accuracy in a mean squared error (MSE) sense for baseband sample vectors, but also that it can provide high-quality approximations for all sample vectors of sinusoids with frequencies in the band of interest. After Section~\ref{sec:review DPSS} provides background on DPSS's, Section~\ref{sec:roast1} provides details on the ROAST construction, fast computations, and theoretical approximation guarantees. The orthogonality of this transform also extends its relevance to new applications, as we describe in Section~\ref{sec:orthapps}. Section~\ref{sec:proof} contains proofs of the main results.  Experiments in Section~\ref{sec:simulations} confirm that ROAST offers signal approximation quality that is comparable to the DPSS, but with a much lower computational burden.

\subsection{DPSS bases}
\label{sec:review DPSS}

To begin, we briefly review some important definitions and properties of DPSS's.

\subsubsection{Definitions}

For any $W\in(0,\frac{1}{2})$, let $\calB_W:\ell_2(\Z)\rightarrow \ell_2(\Z)$ denote a bandlimiting operator that bandlimits the discrete-time Fourier transform (DTFT) of a discrete-time signal to the frequency range $[-W,W]$ (and returns the corresponding signal in the time domain). In addition, for any $N\in\N$, let $\calT_N:\ell_2(\Z)\rightarrow \ell_2(\Z)$ denote the timelimiting operator that zeros out all entries outside the index range $\{0,1,\ldots,N-1\}$.

\begin{definition}(DPSS's \cite{Slepian78DPSS}) Given $W\in(0,\frac{1}{2})$ and $N\in\N$, the Discrete Prolate Spheroidal Sequences (DPSS's) $\{s_{N,W}^{(0)},s_{N,W}^{(1)},\dots,s_{N,W}^{(N-1)}\}$ are real-valued discrete-time sequences that satisfy $\calB_{W}(\calT_N(s_{N,W}^{(\ell)}))=\lambda_{N,W}^{(\ell)}s_{N,W}^{(\ell)}$ for all $l\in\{0,\ldots,N-1\}$. Here $\lambda_{N,W}^{(0)},\dots,\lambda_{N,W}^{(N-1)}$ are the eigenvalues of the operator $\calB_{[-W,W]}\mathcal{T}_N$ with order $1>\lambda_{N,W}^{(0)}>\lambda_{N,W}^{(1)}>\dots>\lambda_{N,W}^{(N-1)}>0$.
\end{definition}

\begin{definition}(DPSS vectors \cite{Slepian78DPSS}) Given $W\in(0,\frac{1}{2})$ and $N\in\N$, the DPSS vectors\footnote{Throughout the paper, finite-dimensional vectors and matrices are indicated by bold characters, while the other variables such as infinite-length sequenes are not in bold typeface.} $\vs_{N,W}^{(0)}$ $\vs_{N,W}^{(1)},\dots,~\vs_{N,W}^{(N-1)}\in\R^N$ are defined by limiting the DPSS's to the index range $\{0,1,\ldots,N-1\}$
	and satisfy
	$$
	\mB_{N,W}\vs_{N,W}^{(\ell)}=\lambda_{N,W}^{(\ell)}\vs_{N,W}^{(\ell)},
	$$
	where $\mB_{N,W}\in\C^{N\times N}$ is the prolate matrix with elements
	\[
	\mB_{N,W}[m,n] = \frac{\sin\left(2\pi W(m-n)\right)}{\pi (m-n)}.
	\]
\end{definition}
Let $\mS_{N,W}$ denote an $N\times N$ matrix whose $\ell$-th column is the DPSS vector $\vs_{N,W}^{(\ell)}$ for all $\ell=0,\ldots,N-1$ and $\mLambda_{N,W}$ be an $N\times N$ diagonal matrix with diagonal entries  being the DPSS eigenvalues $\lambda_{N,W}^{(0)},\ldots,\lambda_{N,W}^{(N-1)}$. The prolate matrix $\mB_{N,W}$ can be factorized as
\[
\mB_{N,W} = \mS_{N,W}\mLambda_{N,W}\mS_{N,W}^*,
\]
which is an eigendecompostion of $\mB_{N,W}$. Here $\mA^*$ represents the adjoint of $\mA$.
The DPSS's are orthogonal on $\mathbb{Z}$ and on $\{0,\dots,N-1\}$, and they are normalized so that
\begin{align*}
\langle\mathcal{T}_N(s_{N,W}^{(k)}),\mathcal{T}_N(s_{N,W}^{(\ell)})\rangle = \begin{cases} 1, & k = \ell, \\ 0, & k \neq \ell. \end{cases}
\end{align*}
Consequently, it can be shown~\cite{Slepian78DPSS} that $\|s_{N,W}^{(\ell)}\|_2^2 = \frac{1}{\lambda_{N,W}^{(\ell)}}$. Thus, when $\lambda_{N,W}^{(\ell)}$ is close to $1$, the corresponding DPSS vector $\vs_{N,W}^{(\ell)}$ has energy mostly concentrated in the frequency range $[-W,W]$. On the other hand when $\lambda_{N,W}^{(\ell)}$ is close to $0$, the corresponding DPSS vector $\vs_{N,W}^{(\ell)}$ has most of its energy outside the frequency range $[-W,W]$. These properties, along with the following result on the distribution of the eigenvalues $\lambda_{N,W}^{(\ell)}$, make the DPSS's a suitable basis to provide a compact representation for sampled bandlimited signals.
\begin{thm}\label{thm:concentration of the DPSS eigenvalues}(Concentration of the spectrum\cite{Slepian78DPSS,DavenportWakin2012CSDPSS,Karnik2016FAST,ZhuWakin2015MDPSS}.)
	For any $W\in(0,\frac{1}{2})$, $N\in\N$, and $\epsilon\in(0,\frac{1}{2})$, we have
	\[
	\lambda_{N,W}^{(\lfloor 2NW \rfloor -1)} \geq \frac{1}{2} \geq \lambda_{N,W}^{(\lceil 2NW \rceil)}
	\]
	and
	\[
	\#\{\epsilon\leq \lambda_{N,W}^{(\ell)}\leq 1-\epsilon\} \leq 2C_N \log \left( \frac{15}{\eps} \right),
	\]
	where $C_N = \frac{4}{\pi^2} \log(8N) + 6$.
\end{thm}
Here $\lfloor a \rfloor $ denotes the largest integer that is not greater  than $a$ and $\lceil a \rceil $ denotes the smallest integer that is not smaller than $a$. Theorem~\ref{thm:concentration of the DPSS eigenvalues} implies that the first $\approx 2NW$ eigenvalues tend to cluster very close to $1$, while the remaining eigenvalues tend to cluster very close to 0, after a narrow transition of width $O(\log(N)\log(\frac{1}{\epsilon}))$.

\subsubsection{Representations of sampled sinusoids and oversampled bandlimited signals}

Define
$$\ve_f:=\begin{bmatrix}e^{j2\pi f 0} \\ e^{j2\pi f 1} \\ \vdots \\ e^{j2\pi f (N-1)}\end{bmatrix} \in \mathbb{C}^{N}$$
for all $f\in [-\frac{1}{2},\frac{1}{2}]$ as the sampled complex exponentials. For any integer $K\in\{1,2,\ldots,N\}$, let $\mS_K :=[\mS_{N,W}]_K$ denote the $N\times K$ matrix formed by taking the first $K$ DPSS vectors (where $N$ and $W$ are clear from the context and typically $K\approx 2NW$). Note that for any orthonormal matrix $\mQ\in\C^{N\times K}$,
\begin{equation}\label{eq:mse with trace}\begin{split}
\int_{-W}^W\left\|\ve_f - \mQ\mQ^*\ve_f\right\|_2^2df &= \int_{-W}^W \trace\left(\ve_f\ve_f^* - \mQ\mQ^*\ve_f\ve_f^*\right)df \\ &= \trace\left(\mB_{N,W} - \mQ\mQ^*\mB_{N,W} \right).
\end{split}\end{equation}
For any value of $K$, the quantity in~\eqref{eq:mse with trace} is minimized by the choice of $\mQ = \mS_K$. This implies that $\mS_K$ is the best basis of $K$ columns to represent (in an MSE sense) the collection of sampled sinusoids $\{\ve_f\}_{f\in[-W,W]}$.  Formally,
\begin{align}\label{eq:redidual by DPSS}
\int_{-W}^{W}\left\|\ve_f-\mS_K\mS_K^*\ve_f\right\|_2^2df =  \sum_{\ell=K}^{N-1}\lambda_{N,W}^{(\ell)},
\end{align}
whereas for each $f\in[-W,W]$, $\|\ve_f\|_2^2 = N$.
It follows from Theorem~\ref{thm:concentration of the DPSS eigenvalues} that $\mS_K$ provides very accurate approximations (in an MSE sense) for all sampled sinusoids $\{\ve_f\}_{f\in[-W,W]}$ if one chooses $K$ slightly larger than $2NW$. We note that this efficiency is in contrast to the DFT, where certain ``on-grid'' sinusoids (those whose frequencies are harmonic multiples of $1/N$) can be represented using just one DFT basis vector, but all other ``off-grid'' sinusoids require $O(N)$ DFT basis vectors due to frequency leakage.


We note that any representation guarantee  for sampled sinusoids $\{\ve_f\}_{f\in[-W,W]}$
can also be used  for finite-length sample vectors arising from sampling random bandlimited baseband signals.  Suppose $x$ is a continuous-time, zero-mean, wide sense stationary random process with power spectrum
$$P_{x}(F)=\left\{\begin{array}{ll} \frac{1}{B_{\textup{band}}},& F\in[-\frac{B_{\textup{band}}}{2},\frac{B_{\textup{band}}}{2}], \\0, & \textup{otherwise}.\end{array}\right.$$
Let $\vx = [x(0) ~x(T_s) ~ \cdots ~x((N-1)T_s)]^T\in \C^N$ denote a finite vector of samples acquired from $x(t)$ with a sampling interval of $T_s\leq 1/B_\textup{band}$.
Let $f_c = F_cT_s$ and $W = \frac{B_{\textup{band}}T_s}{2}$. We have~\cite{DavenportWakin2012CSDPSS}
\begin{align}\label{eq:mse for sampled bandlimited signal}
\mathbb E\left[\left\|\vx-\mQ\mQ^*\vx\right\|_2^2\right] =  \frac{1}{2W}\int_{-W}^{W}\left\|\ve_f-\mQ\mQ^*\ve_f\right\|_2^2df.
\end{align}

Finally, let $\mF_{N,W}$ denote the partial normalized DFT matrix with the lowest $2\lfloor NW\rfloor+1$ frequency DFT vectors of length $N$, i.e.,
\[
\mF_{N,W} = \frac{1}{\sqrt N}\left[\begin{array}{ccc}\ve_{-\frac{\lfloor NW\rfloor}{N}}  & \cdots & \ve_{\frac{\lfloor NW\rfloor}{N}} \end{array}\right].
\]
It follows that $\mF_{N,W} \mF_{N,W}^*$ is an orthogonal projector onto the column space of $\mF_{N,W}$. The following result states that the difference between the prolate matrix $\mB_{N,W}$ and $\mF_{N,W} \mF_{N,W}^*$ is effectively low rank.

\begin{thm}\cite{Karnik2016FAST} \label{thm:prolateFFTLR}
	Let $N \in \N$ and $W \in (0, \tfrac{1}{2})$ be given.  Then for any $\eps \in (0,\tfrac{1}{2})$, there exist $N\times N$ matrices $\mL$ and $\mE$ such that
	\[
	\mB_{N,W} = \mF_{N,W} \mF_{N,W}^* + \mL + \mE,
	\]
	where
	\[
	\rank(\mL) \le C_N \log \left( \frac{15}{\eps} \right), \quad \| \mE \| \le \eps.\]
	Here $C_N$ is the constant specified in Theorem~\ref{thm:concentration of the DPSS eigenvalues}.
\end{thm}
This result is a key factor in fast computing an approximate Slepian transform in~\cite{Karnik2016FAST} and will play an important role in the construction of the ROAST, which can be used for computing fast orthogonal approximations of sampled sinusoids and bandlimited signals. 

\subsection{ROAST: Rapid Orthogonal Approximate Slepian Transform}
\label{sec:roast1}

\subsubsection{Construction and relation to the DPSS subspace}

In~\cite{Karnik2016FAST}, we  demonstrated a fast method to approximately project an arbitrary vector onto the subspace spanned by the first slightly more than $2NW$ eigenvectors of $\mB_{N,W}$ (i.e., the DPSS vectors) by utilizing the fact that the difference between $\mB_{N,W}$ and $\mF_{N,W} \mF_{N,W}^*$ approximately has a rank of $O(\log N)$ (see Theorem~\ref{thm:prolateFFTLR}). Note that, in~\cite{Karnik2016FAST}, the approximate  projection is not a true orthogonal projection onto any subspace. Here, we exhibit a subspace that captures most of the energy in the first $2NW$ DPSS vectors (and also the energy in sampled sinusoids within the band of interest), and this subspace has an orthogonal projector that can be applied efficiently to an arbitrary vector.

By utilizing the result that $\mB_{N,W}-\mF_{N,W} \mF_{N,W}^*$ is approximately low rank and also that $\mF_{N,W}$ can be applied to a vector efficiently with the FFT, we build  an orthonormal basis for our subspace by concatenating $\mF_{N,W}$ with a certain matrix $\mQ'$ as follows:
\begin{align*}
\mQ = \left[\mF_{N,W} \quad \mQ'  \right],
\end{align*}
where $\mQ'$ is an $N\times R$ (for some $R$ that we can choose as desired) orthonormal matrix that is also orthogonal to $\mF_{N,W}$. Let $\overline{\mF}_{N,W}$ denote the $N\times (N-2\lfloor NW\rfloor -1)$ matrix with the highest frequency $N-2\lfloor NW\rfloor -1$ DFT vectors of length $N$. Thus $\mF_N := \left[\mF_{N,W}\quad\overline{\mF}_{N,W} \right]$ is the normalized $N\times N$ DFT matrix. Since $\mQ'$ must be orthogonal to $\mF_{N,W}$ and the columns of $\mQ'$ must be orthonormal, we can write $\mQ'$ as $\mQ' = \overline{\mF}_{N,W} \mV$, for some $\mV\in \C^{(N-2\lfloor NW\rfloor -1)\times R}$ that is orthonormal (one can verify that $\mF_{N,W}^*\mQ' = \mzero$ and $(\mQ')^*\mQ' = \mId$). Thus, the desired orthogonal approximate Slepian basis is given as
\begin{align}\mQ = \left[\mF_{N,W} \quad \overline{\mF}_{N,W} \mV  \right], \quad \mV^\T\mV = \mId.
\label{eq:form of Q}\end{align}

The optimal $\mV$ is chosen such that the subspace spanned by $\mQ$ captures the important DPSS vectors. (Since all the DPSS vectors $\vs_{N,W}^{(0)}, \ldots, \vs_{N,W}^{(N-1)}$ form an orthonormal basis for $\C^N$, no subspace of $\C^N$ can capture all of them except $\C^N$ itself.) To illustrate how we obtain $\mV$, consider the following weighted least squares problem
\begin{align}
\minimize_{\mQ} \varrho(\mQ):=\sum_{\ell = 0}^{N-1}\lambda_{N,W}^{(\ell)}\left\|\vs_{N,W}^{(\ell)} -\mQ\mQ^*\vs_{N,W}^{(\ell)}\right\|_2^2.
\label{eq:define rho Q}\end{align}
Here we use the DPSS eigenvalue $\lambda_{N,W}^{(\ell)}$ to weight the energy in the DPSS vector $\vs_{N,W}^{(\ell)}$ that is not captured by $\mQ$. The reason is that the larger the DPSS eigenvalue, the more concentration the corresponding DPSS vector has in the frequency domain, implying that the DPSS vector is more important in practical applications such as representing sampled bandlimited signals (see~\eqref{eq:redidual by DPSS}).  To solve~\eqref{eq:define rho Q}, we rewrite $\varrho(\mQ)$ as
\begin{equation}\label{eq:rho Q to MSE}\begin{split}
\varrho(\mQ) &= \trace\bigg(\sum_{\ell = 0}^{N-1}\lambda_{N,W}^{(\ell)}\vs_{N,W}^{(\ell)} (\vs_{N,W}^{(\ell)})^\T -\mQ\mQ^* \lambda_{N,W}^{(\ell)}\vs_{N,W}^{(\ell)}(\vs_{N,W}^{(\ell)})^\T\bigg)\\
& = \trace\left(\mB_{N,W} - \mQ\mQ^*\mB_{N,W} \right)\\
& = \int_{-W}^W\left\|\ve_f - \mQ\mQ^*\ve_f\right\|_2^2df,
\end{split}\end{equation}
where the last line follows from~\eqref{eq:mse with trace}. In other words, an orthonormal basis $\mQ$ obtained by minimizing $\varrho(\mQ)$ is also an optimal basis to represent sampled sinusoids (and thus also certain bandlimited signals) in the MSE sense.

Plugging $\mQ = \left[\mF_{N,W} \quad \overline{\mF}_{N,W} \mV  \right]$ into the above equation yields
\begin{align*}
\varrho(\mQ)= \trace(\overline{\mF}_{N,W}^* \mB_{N,W}\overline{\mF}_{N,W}  - \mV \mV^* \overline{\mF}_{N,W}^* \mB_{N,W}\overline{\mF}_{N,W}),
\end{align*}
which suggests that setting $\mV$ equal to the $R$ dominant left singular vectors of $\overline{\mF}_{N,W}^*\mB_{N,W}\overline{\mF}_{N,W}$ results in a relatively small representation residual $\varrho(Q)$ as long as $\overline{\mF}_{N,W}^*\mB_{N,W}\overline{\mF}_{N,W}$ has an effective rank of $R$. In fact, we find that certain numerical issues can be avoided by adopting the $R$ dominant left singular vectors of $\overline{\mF}_{N,W}^*\mB_{N,W}$ (rather than $\overline{\mF}_{N,W}^*\mB_{N,W}\overline{\mF}_{N,W}$), and that the same strong theoretical guarantees can be established for this construction. The following result provides such a guarantee for the standard ROAST construction involving the singular vectors of $\overline{\mF}_{N,W}^*\mB_{N,W}$; we briefly revisit the idea of a constructing involving the singular vectors of $\overline{\mF}_{N,W}^*\mB_{N,W}\overline{\mF}_{N,W}$ in Section~\ref{sec: proof average repres error}.

\begin{thm}\label{thm:subspace relationship between S_K and Q}(Representation guarantee for DPSS vectors)
	Fix $N \in \N$ and $W \in (0, \tfrac{1}{2})$. For any $\eps \in (0,\tfrac{1}{2})$, fix $K$ to be such that $\lambda_{N,W}^{(K-1)}\geq\epsilon$ and  set $R = \lceil C_N\log \left( 15/\epsilon \right)\rceil$, where $C_N$ is the constant specified in Theorem~\ref{thm:concentration of the DPSS eigenvalues}. Then the orthonormal basis $\mQ = \left[\mF_{N,W} \quad \overline{\mF}_{N,W} \mV  \right]$ with $\mV\in\C^{(N-2\lfloor NW\rfloor -1)\times R}$ containing the $R$ dominant left singular vectors of $\overline{\mF}_{N,W}^*\mB_{N,W}$ satisfies
	\begin{align*}
	&\|\mS_K\mS_K^* - \mQ\mQ^*\mS_K \mS_K^*\|^2 \leq \epsilon,\\
	&\| \vs_{N,W}^{(\ell)} - \mQ \mQ^*\vs_{N,W}^{(\ell)}  \|_2^2\leq \epsilon,
	\end{align*}
	for all $\ell = 0,1,\ldots, K-1$. By slightly increasing $R$ to $R = \lceil C_N\log \left( 15N/\epsilon \right)\rceil$, the subspace angle $\Theta_{\mS_K,\mQ}$ between the columns spaces of $\mS_K$ and $\mQ$ satisfies
	\begin{align*}
	\cos\left(\Theta_{\mS_K,\mQ}\right) \geq \sqrt{1- \eps}.
	\end{align*}
\end{thm}
The formal definition of (the largest principal) angle between two subspaces is given in Definition~\ref{def:subspace angle}. Informally, if the subspace angle $\Theta$ is small, the two subspaces are nearly linearly dependent and one subspace is almost ``contained'' in the other subspace. Here, to guarantee that the column space of $\mS_K$ is almost ``contained'' in the column space of $\mQ$, one can make $\Theta_{\mS_K,\mQ}$ arbitrary small by increasing $R$. However, we note that we are not guaranteed that $\|\mQ\mQ^* - \mS_K \mS_K^*\|$ is small since in general $\|\mQ\mQ^* - \mS_K \mS_K^*\| =1$ if $\mQ$ and $\mS_K$ have a different number of columns. Instead, we are guaranteed that the subspace spanned by the columns of $\mS_K$ is approximately within the column space of $\mQ$ and the angle between the two subspaces is small by Theorem~\ref{thm:subspace relationship between S_K and Q}. We also note that the bound on $\|\mS_K\mS_K^* - \mQ\mQ^*\mS_K\mS_K^*\|$ is useful since  for any vector $\va\in\C^N$
\begin{align*}
\|\va - \mQ \mQ^*\va\|_2 &\leq  \|\va - \mQ \mQ^* \mS_K \mS_K^*\va\|_2\\
&\leq  \|\va -\mS_K \mS_K^*\va\|_2 + \|\mS_K\mS_K^* - \mQ\mQ^*\mS_K\mS_K^*\|_2\|\va\|_2\\
&\leq \|\va -\mS_K \mS_K^*\va\|_2 + \sqrt{\epsilon} \| \va\|_2,
\end{align*}
which\footnote{Here the first inequality holds because $\mQ \mQ^*\va$ is the orthogonal projection of $\va$ onto $\Span(\mQ)$ (the column space of $\mQ$) and thus is closest to $\va$ among all points in $\Span(\mQ)$, in which $\mQ \mQ^* \mS_K \mS_K^*\va$ also lies.} implies any representation guarantee for $\mS_K$ can be utilized for $\mQ$.

\subsubsection{Representations of sampled sinusoids and oversampled bandlimited signals}

As illustrated in~\eqref{eq:rho Q to MSE}, the orthonormal matrix obtained by minimizing $\varrho(\mQ)$ is also expected to accurately represent sampled sinusoids within the band of interest in the MSE sense.  This is formally established in the following results.
\begin{thm}\label{thm:average repres error}(Average representation error) Fix $W\in(0,\frac{1}{2})$ and $N\in\N$. For any $\eps \in (0,\tfrac{1}{2})$, set
	\[ R = \max\left\{ \left\lceil C_N \log \left( \frac{15 C_N }{N\eps} \right)\right\rceil + 1, 0\right\},\]
	where $C_N$ is the constant specified in Theorem~\ref{thm:concentration of the DPSS eigenvalues}. Then  the orthonormal basis $\mQ = \left[\mF_{N,W} \quad \overline{\mF}_{N,W} \mV  \right]$ with $\mV\in\C^{(N-2\lfloor NW\rfloor -1)\times R}$ containing the $R$ dominant left singular vectors of $\overline{\mF}_{N,W}^*\mB_{N,W}$ satisfies
	\[
	\int_{-W}^W\frac{\left\|\ve_f - \mQ\mQ^*\ve_f\right\|_2^2}{\|\ve_f\|_2^2}df \leq \epsilon
	\]
	
\end{thm}
A similar approximation guarantee can be established for vectors arising from sampling random bandlimited signals by using~\eqref{eq:mse for sampled bandlimited signal}.

In~\cite{ZhuWakin2015MDPSS}, we rigorously show that every discrete-time sinusoid with a frequency $f\in[-W,W]$ is well-approximated by the DPSS basis $\mS_K$ when $K$ is slightly larger than $2NW$. The proof is based on an asymptotic result on the DTFT of the DPSS basis functions (which are known as discrete prolate spheroidal wave functions (DPSWF's)) and the result is thus asymptotic. Here we use a different approach to obtain a non-asymptotic guarantee for approximating every discrete-time sinusoid with a frequency $f\in[-W,W]$. Noting that $\left\|\ve_f - \mQ\mQ^*\ve_f\right\|_2^2$ is differentiable everywhere, we first show that its derivative is bounded above by $2\pi N^2$. Then by utilizing the previous result on $\int_{-W}^W\left\|\ve_f - \mQ\mQ^*\ve_f\right\|_2^2df$, one obtains a similar bound on $\left\|\ve_f - \mQ\mQ^*\ve_f\right\|_2^2$.

\begin{thm}\label{thm:fast for sinusoid}(Representation guarantee for pure sinusoids)
	Let $N \in \N$ and $W \in (0, \tfrac{1}{2})$ be given such that $W\geq \frac{1}{4\pi N}$.  For any $\eps \in (0,\tfrac{1}{2})$, set
	\[
	R =  \max\left(\left\lceil C_N \log( \frac{60\pi C_N}{\eps^2} )\right\rceil +1,~ \left\lceil C_N \log ( \frac{15 C_N}{NW\eps})\right\rceil +1 \right),\]
	where $C_N$ is the constant specified in Theorem~\ref{thm:concentration of the DPSS eigenvalues}. Then the orthonormal basis $\mQ = \left[\mF_{N,W} \quad \overline{\mF}_{N,W} \mV  \right]$ with $\mV\in\C^{(N-2\lfloor NW\rfloor -1)\times R}$ containing the $R$ dominant left singular vectors of $\overline{\mF}_{N,W}^*\mB_{N,W}$ satisfies
	\[
	\frac{\left\|\ve_f - \mQ\mQ^*\ve_f\right\|_2^2}{\|\ve_f\|_2^2} \leq \epsilon
	\]
	for all $f\in[-W,W]$.
\end{thm}
\begin{remark}
	In~\cite{ZhuWakin2015MDPSS}, we show a similar but asymptotic result for the Slepian basis as follows. Fix $W \in (0, \tfrac{1}{2})$ and $\delta\in(0,\frac{1}{2W}-1)$. Let $K = 2NW(1+\delta)$. Then there exist constants $\widetilde C_1,\widetilde C_2$ and $N_0\in\N$ (which may depend on $W$ and $\delta$) such that
	\[
	\frac{\left\|\ve_f - \mS_K\mS_K^*\ve_f\right\|_2^2}{\|\ve_f\|_2^2} \leq \widetilde C_1 N^{3/2}e^{-\widetilde C_2 N}
	\]
	for all $N\geq N_0$ and $f\in[-W,W]$. Compared with this result, Theorem~\ref{thm:fast for sinusoid} is non-asymptotic and provides  detail on the constants involved.
\end{remark}

Finally, we remark that for $\mQ = \left[\mF_{N,W} \quad \overline{\mF}_{N,W} \mV  \right]$ with $\mV\in\C^{(N-2\lfloor NW\rfloor -1)\times R}$, both $\mQ$ and $\mQ^*$ can be applied to a vector with computational complexity $O(N\log N + NR)$. As an example, for any $\va\in\C^N$, $\widetilde\va = \left[\mF_{N,W} \quad \overline{\mF}_{N,W} \right]^\H \va$ can be efficiently computed using the FFT with  complexity  $O(N\log N)$. Then $\mV^*\widetilde \va_2$ can be computed via conventional matrix-vector multiplication with  complexity  $O(NR)$, where $\widetilde \va_2$ is the sub-vector obtained by taking the last $N-2\lfloor NW\rfloor -1$ entries of $\widetilde\va_2$. Thus the total computational complexity for computing $\mQ^*\va$ is $O(N\log N + NR)$.

\subsubsection{ROAST construction with a randomized algorithm}
\label{sec:roast2}

We note that the DPSS vectors are not involved  in constructing $\mV$ and $\mQ$. Directly computing $\mV$ with the Businger-Golub algorithm~\cite{businger1969algorithm} has complexity $O(N(N-2\lfloor NW\rfloor -1)R)$. Noting that $\overline{\mF}_{N,W}^*\mB_{N,W}$ is effectively low rank, however, we can apply a fast randomized algorithm~\cite{halkoRandomizedAlgorithm} to compute an approximate basis for the range of $\overline{\mF}_{N,W}^*\mB_{N,W}$.  Let $\mOmega$ be an $N\times P$ standard Gaussian matrix. We construct a matrix $\overline\mV$ whose columns form an orthonormal basis for the range of $\overline{\mF}_{N,W}^*\mB_{N,W}\mOmega$. By applying the FFT, the complexity of computing $\overline{\mF}_{N,W}^*\mB_{N,W} \mOmega$ is $O(PN\log N)$. Computing an orthonormal basis for the range of $\overline{\mF}_{N,W}^*\mB_{N,W} \mOmega$ requires $O(NP^2)$ flops.
The following results establish the dimensionality of $\overline\mV$ needed and the representation guarantee with the corresponding basis.

\begin{thm}\label{thm:randomized algorithm for finding V}(Guarantee for randomized algorithm)
	Fix $N \in \N$ and $W \in (0, \tfrac{1}{2})$. For any $\eps \in (0,\tfrac{1}{2})$, fix $K$ to be such that $\lambda_{N,W}^{(K-1)}\geq\epsilon$. Let $\mOmega$ be an $N\times P$ standard Gaussian matrix, with $P$ specified as below. Also let $\overline\mV$ be an orthonormal basis for the column space of the sample matrix $\overline{\mF}_{N,W}^*\mB_{N,W}\mOmega$. Then the orthonormal basis $\mQ = \left[\mF_{N,W} \quad \overline{\mF}_{N,W} \overline\mV  \right]$ has the following expression ability in expectation.
	\begin{enumerate}[(i)]
		\item Setting
		\[
		P = \left\lceil 2C_N \log \left( \frac{30+15e}{\eps} \right)\right\rceil +3,
		\]
		we are guaranteed that
		\begin{align*}
		&\mathbb{E}\left[ \left\|\mS_K\mS_K^* - \mQ\mQ^*\mS_K \mS_K^*\right\|^2\right] \leq \epsilon,\\
		&\mathbb{E} \left[\left \| \vs_{N,W}^{(\ell)} - \mQ \mQ^*\vs_{N,W}^{(\ell)}\right\|_2^2\right] \leq \epsilon
		\end{align*}
		for all $l = 0,1,\ldots, K-1$. By slightly increasing $P$ to
		\[P = \left\lceil 2C_N \log \left( \frac{\left(30+15e\right)N}{\eps} \right)\right\rceil+3,\]
		we have
		\[
		\mathbb{E} \left[\cos\left(\Theta_{\mS_K,\mQ}\right)\right]  \geq \sqrt{1 - N\eps}.
		\]
		
		\item Sampled sinusoids within the band of interest are well-approximated by $\mQ$ in expectation:
		\[
		\mathbb{E} \left[\int_{-W}^W\frac{\left\|\ve_f - \mQ\mQ^*\ve_f\right\|_2^2}{\|\ve_f\|_2^2}df \right] \leq \epsilon
		\]
		with
		\[ P = \left\lceil\frac{4}{3}C_N \log \left( \frac{15\sqrt{2C_N}}{\eps} \right) + \frac{7}{3}\right\rceil.
		\]
		\item The orthonormal basis $\mQ$ can also  capture most of the energy in each pure sinusoid:
		\[
		\mathbb{E} \left[\frac{\left\|\ve_f - \mQ\mQ^*\ve_f\right\|_2^2}{\|\ve_f\|_2^2}\right] \leq \epsilon
		\]
		for all $f\in[-W,W]$ with
		\begin{align*}
		P = \max\bigg(&\left\lceil\frac{4}{3}C_N \log ( \frac{
			60\pi N\sqrt{2C_N}}{\eps^2} )+\frac{7}{3}\right\rceil,\left\lceil \frac{4}{3}C_N \log ( \frac{
			15\pi \sqrt{2C_N }}{W\eps} )+\frac{7}{3}\right\rceil\bigg).
		\end{align*}
	\end{enumerate}
	Here $\mathbb{E}$ denotes expectation with respect to the random matrix $\mOmega$.
\end{thm}
\begin{remark}
	Using concentration of measure~\cite{halkoRandomizedAlgorithm}, we can argue that the results above hold for a particular sampling matrix $\mOmega$ with high probability.
\end{remark}

In summary, the ROAST offers a computationally efficient alternative to the DPSS with virtually the same approximation performance. The ROAST could therefore be considered for use in many of the applications involving DPSS's that were described earlier in this introduction. For example, in through-the-wall radar imaging using stepped-frequency SAR~\cite{AhmadQianAmin2015WallCluterDPSS,Zhu2015targetDetectDPSS}, the wall return is modeled as a sampled bandpass signal and thus the ROAST can be used to efficiently mitigate the wall return at each antenna.

\subsection{Benefits of an orthonormal basis}
\label{sec:orthapps}

For any $\delta \in (0,\tfrac{1}{2})$, fix $K$ to be such that $\lambda_{N,W}^{(K-1)}\geq\delta$. In~\cite{Karnik2016FAST}, we demonstrated a fast factorization of $\mS_K\mS_K^*$ by constructing $\mT_1 = \begin{bmatrix}\mF_{N,W} & \mD_1  \end{bmatrix}$ and $\mT_2 = \begin{bmatrix}\mF_{N,W} & \mD_2 \end{bmatrix}$ (with $\mD_1,\mD_2\in\R^{N\times r}$) such that
\begin{align}
\|\mS_K\mS_K^* - \mT_1\mT_2^*\|\leq 2\delta , \ \text{with} \ r \leq 3 C_N \log \left( \frac{15}{\delta} \right).
\label{eq:FST}\end{align}
We utilize FST to denote the approximate projection  $\mT_1\mT_2^*$.

However, neither $\mT_1$ nor $\mT_2$ is orthonormal and in general $\|\mT_2^*\vx\| \neq \|\mT_1\mT_2^* \vx\|$ and $\|\mT_2^*\vx\| \neq \|\mT_2\mT_2^* \vx\|$. Moreover, neither $\mT_1$ nor $\mT_2$ is well conditioned (i.e., both have a large condition number). In some applications like orthogonal precoding for wireless communication~\cite{zemen2017orthogonal}, an orthonormal transform $\mQ$ is required or preferred, in order to ensure that $\left\|\mP_{\mQ}\vx\right\| = \left\|\mQ^*\vx\right\|$ or that $\mQ$ is well conditioned. We list two more stylized applications in signal processing below.

\subsubsection{Signal recovery}

Suppose $\vx\in \C^{N}$ is a sampled bandlimited signal with digital frequencies within the band $[-W,W]$ and we observe it through
\[
\vy = \mPhi \vx,
\]
where $\mPhi \in \C^{M\times N}$ ($2NW\leq M\leq N$) is the sensing matrix. Knowing that $\vx$ approximately lives in the subspace spanned by $\mS_K$, we recover $\vx$ by solving
\[
\minimize_{\valpha} \|\vy - \mPhi \mS_K \valpha\|_2^2,
\]
which is also a key part in a compressive sensing recovering algorithm for multiband analog signals~\cite{DavenportWakin2012CSDPSS} (see also~\cite{wakin2012nonuniform}). The above least-squares problem is equivalent to the following system of linear equations
\begin{align}
\mS_K^*\mPhi^*\mPhi\mS_K \valpha =  \mS_K^*\mPhi^*\vy,
\label{eq:linear system}\end{align}
which can be solved by numerical algorithms such as conjugate gradient descent (CGD)~\cite{saad2003iterative}. The computational complexity of the CGD method depends on two factors: the convergence speed, which depends on the condition number of the system $\mA:=\mS_K^*\mPhi^*\mPhi\mS_K$ and determines the number of iterations required, and the computational burden in each iteration, mainly involving the application of $\mA$ to a length-$M$ vector. Utilizing a structured sensing matrix $\mPhi$ that has a fast implementation (such as the fast Johnson-Lindenstrauss transform~\cite{ailon2009fast}), we can efficiently implement $\mA$ if we replace $\mS_K$ by the fast transform $\mT_1$ or $\mT_2$~\cite{Karnik2016FAST} or the ROAST $\mQ$ of the form~\eqref{eq:form of Q}. Unfortunately, both $\mT_1$ and $\mT_2$ have large condition number, resulting in slow convergence of the CGD method since the corresponding system $\mA$ in general also has large condition number. Thus, in this case, the orthonormal basis $\mQ$ is preferable.

\subsubsection{Line spectral estimation}

Consider a measurement vector $\vy$ consisting of a superposition of $r$ sampled exponentials:
\[
\vy = \sum_{i=1}^r \alpha_i^\star \ve_{f_i^\star},
\]
where $\{f_i^\star\}$ are the frequencies and $\{\alpha_i^\star\}$ are the corresponding coefficients.
We may attempt to recover the frequencies $\{f_1^\star,\ldots,f_r^\star\}$ by solving the following nonlinear least squares problem
\begin{align}
\{\widehat f_i,\widehat \alpha_i\} := \arg\min_{f_i,\alpha_i} \left\|\vy -  \sum_{i=1}^r \alpha_i \ve_{f_i}\right\|_2^2.
\label{eq:full prob}\end{align}
Suppose we are given a priori knowledge that the frequencies $f_i^*\in[-W,W]$ for all $i\in\{1,\ldots,r\}$. Then we can reduce the computational cost of solving by~\eqref{eq:full prob} by projecting the measurements $\vy$ onto the range space of $\mQ$ \cite{hokanson2015}:
\begin{equation}\begin{split}
\{\overline f_i,\overline\alpha_i\} &:=\arg\min_{f_i,\alpha_i} \left\|\mP_{\mQ}\left( \vy -  \sum_{i=1}^r \alpha_i \ve_{f_i}\right)\right\|_2^2 = \arg\min_{f_i,\alpha_i} \left\|\mQ^*\left( \vy -  \sum_{i=1}^r \alpha_i \ve_{f_i}\right)\right\|_2^2.
\end{split}\label{eq:projected prob}\end{equation}
It is shown in \cite{hokanson2015} that the projected problem~\eqref{eq:projected prob} has the same stationary points as the full problem~\eqref{eq:full prob} under certain conditions on the range space of $\mQ$. When applying an optimization method like Gauss-Newton, the advantage of the projected problem~\eqref{eq:projected prob} over the full problem~\eqref{eq:full prob} is that each optimization step is much cheaper since the projected Jacobian has much smaller size.

Based on this observation, for the general case where the frequencies  lie in multiple bands, \cite{hokanson2015} provides an iterative algorithm that in each iteration finds one underlying band and projects the signal onto this band, then applies Gauss-Newton to solve the projected problem. We also note that our $\mQ$ can be further reduce the computational cost in \cite{hokanson2015} since $\mQ$ can be efficiently applied to a vector, while the orthonormal basis utilized in \cite{hokanson2015} is a numerical approximation (obtained by performing PCA on a set of sinusoids) to the Slepian basis $\mS_K$.

\subsection{Comparision of ROAST and FST~\cite{Karnik2016FAST}}\label{sec:comparision}

There are some similarities and differences between ROAST and FST (i.e., $\mT_1\mT_2^*$ in \eqref{eq:FST}). With respect to the similarities, both ROAST and FST consist of two parts: the partial DFT matrix (which can be applied to a vector efficiently via the FFT) and a skinny matrix (which can also be efficiently applied to any vector with standard matrix-vector multiplication since the number of columns is $O(\log N)$). Aside from the fact that ROAST corresponds to an orthonormal basis while FST is not an exact orthogonal projection, ROAST and FST also differ in the following respects.

$(i)$ FST explicitly attempts to approximate the operator $\mS_K\mS_K^*$, while ROAST is motivated by the goal of approximating the subspace spanned by the DPSS basis vectors. To better reveal the subtle difference between these two goals, let us take a closer look at the objective function~\eqref{eq:define rho Q} corresponding to ROAST:
\begin{align*}
&\minimize_{\mQ} \sum_{\ell = 0}^{N-1}\lambda_{N,W}^{(\ell)}\left\|\vs_{N,W}^{(\ell)} -\mQ\mQ^*\vs_{N,W}^{(\ell)}\right\|_2^2= \trace\left(\mS\mLambda\mS^* - \mQ\mQ^*\mS\mLambda\mS^* \right).
\end{align*}
In this expression, note that we use the DPSS eigenvalue $\lambda_{N,W}^{(\ell)}$ to weight the energy in the DPSS vector $\vs_{N,W}^{(\ell)}$ that is not captured by $\mQ$. Such an eigenvalue-based weighting (most of the weights are either very close to 1 or 0) is not present in the FST objective $\|\mS_K\mS_K^* - \mT_1\mT_2^*\|$. To see why it may not be appropriate to approximate $\mS_K\mS_K^*$ with $\mQ\mQ^*$, we first note that for two orthogonal projectors $\mP_{\calA}$ and $\mP_\calB$, $\|\mP_\calA - \mP_\calB\| =1$ if the dimension of subspace $\calA$ does not equal the dimension of subspace $\calB$. Therefore, $\|\mS_K\mS_K^* - \mQ\mQ^*\|$ will always equal $1$ unless the number of columns in $\mQ$ is set exactly equal to $K$. 
Even if we set the number of columns we use for $\mQ$ equal to $K$, let us take a closer look at what form $\|\mS_K\mS_K^* - \mQ\mQ^*\|$ would take if $\mQ$ has the form $\left[\mF_{N,W} \quad \overline{\mF}_{N,W} \mV  \right]$:
\begin{equation}\begin{split}
&\left\|\mS_K\mS_K^* - \mQ\mQ^*\right\|=\left\|\begin{bmatrix}\mF_{N,W}^*\mS_K\mS_K^*\mF_{N,W} - \mId & \mF_{N,W}^*\mS_K\mS_K^*\overline\mF_{N,W}\\ \overline\mF_{N,W}^*\mS_K\mS_K^*\mF_{N,W}  & \overline\mF_{N,W}^*\mS_K\mS_K^*\overline\mF_{N,W} - \mV\mV^*
\end{bmatrix}\right\|.
\end{split}\label{eq:SS - QQ}\end{equation}
In \eqref{eq:SS - QQ}, we see that regardless of the choice of $\mV$ (even if we could make the bottom right block equal to zero), the overall quantity $\left\|\mS_K\mS_K^* - \mQ\mQ^*\right\|$ could be still large since the other three blocks in the right hand side of \eqref{eq:SS - QQ} are not negligible (in terms of the spectral norm), though probably all of them are low-rank.

$(ii)$ Although \cite{Karnik2016FAST} focuses on approximating $\mS_K\mS_K^*$, it is also possible to derive signal approximation guarantees for the FST. However, these will be slightly weaker than those for ROAST in that we require a slightly larger size (number of columns in $\mT_1$ and $\mT_2$) for FST to have a similar approximation guarantee. In particular, using a similar approach to that used for establishing Theorem~\ref{thm:average repres error}, we have
\begin{equation}\begin{split}
&\int_{-W}^W\frac{\left\|\ve_f - \mT_1\mT_2^*\ve_f\right\|_2^2}{\|\ve_f\|_2^2} df  \leq 2\epsilon,\ \text{with} \ r = 3\max\left\{ \left\lceil C_N \log \left( \frac{15 C_N }{N\eps} \right)\right\rceil + 1, 0\right\}.
\end{split}
\label{eq:mse FST}\end{equation}
Comparing \eqref{eq:mse FST} and Theorem~\ref{thm:average repres error}, we see that FST requires a slightly larger size for a comparable approximation guarantee. In practice, we observe that ROAST requires much smaller size than FST to achieve a similar approximation quality (see Section~\ref{sec:simulations}), since ROAST  is constructed by explicitly minimizing the MSE for approximating sampled sinusoids.

$(iii)$ We note that although the approximation guarantee \eqref{eq:mse FST} for FST  and the one in Theorem~\ref{thm:average repres error} for ROAST provide similar upper bounds on the number of columns in the skinny matrices $\mD_1$, $\mD_2$ and $\mV,$ the construction of these matrices is different. For FST, given $K$ and $\delta$, we provided an explicit construction for the skinny matrices $\mD_1,\mD_2\in \R^{N\times r}$ in~\cite{Karnik2016FAST} with an upper bound on $r$ given in~\eqref{eq:FST}. For ROAST, we construct $\mV$ by computing the singular vectors of $\overline{\mF}_{N,W}^*\mB_{N,W}$ and thus there is freedom to choose the number of singular vectors to be utilized.\footnote{Though the number of columns for $\mV$ in Theorems~\ref{thm:average repres error}-\ref{thm:randomized algorithm for finding V} matches the information-theoretical bound, it is still quite conservative compared to  experimental results. The simulation results in Section~\ref{sec:simulations} indicate that ROAST with $R = 4\log N$ gives very accurate representations for most sampled sinusoids and bandlimited signals.} This freedom is useful in applications like orthogonal precoding for wireless communication~\cite{zemen2017orthogonal} where one has a requirement on the size of the transforms (and thus $R$).

We compare the speed and approximation performance of ROAST and FST using numerical experiments in \Cref{sec:simulations}.


\section{Proof of main results}
\label{sec:proof}

\subsection{Supporting results}
We first establish the following definition of angle between subspaces to compare subspaces of possibly different dimensions.
\begin{definition}\label{def:subspace angle}
	Let $\calS_{\mA}$ and $\calS_{\mB}$ be the subspaces formed by the columns of the matrices $\mA$ and $\mB$ respectively. The subspace angle $\Theta_{\mA,\mB}$ between $\mathcal{S}_{\mA}$ and $\mathcal{S}_{\mB}$ is given by
	$$
	\cos(\Theta_{\mA,\mB}):=\inf_{\va\in\mathcal{S}_{\mA}, \left\|\va\right\|_2=1}\left\|\mP_{\mB}\va\right\|_2
	$$
	if $\dim(\mathcal{S}_{\mB})\geq \dim(\mathcal{S}_{\mA})$, or
	$$
	\cos(\Theta_{\mA,\mB}):=\inf_{\vb\in\mathcal{S}_{\mB}, \|\vb\|_2=1}\|\mP_{\mA}\vb\|_2
	$$
	if $\dim(\mathcal{S}_{\mB})<\dim(\mathcal{S}_{\mA})$. Here $\mP_{\mB}$ (or $\mP_{\mA}$) denotes the orthogonal projection onto the column space of $\mB$ (or $\mA$).
\end{definition}
We remark that when the subspaces $\calS_{\mA}$ and $\calS_{\mB}$ have the same dimension, our definition of subspace angle coincides with the {\em subspace gap}~\cite{duggal2004subspace}, defined as $\sin(\Theta_{\mA,\mB})$. Smaller $\Theta_{\mA,\mB}$ indicates a smaller gap between $\calS_{\mA}$ and $\calS_{\mB}$. We also connect our definition of subspace angle to {\em principal angles} between two subspaces defined as follows.

\begin{definition}\label{def:principal angle}
	\cite{bjorck1973numerical} Suppose $\mA\in\R^{N\times p}$ and $\mB\in\R^{N\times q}$ are orthonormal bases for the subspaces $\calS_{\mA}\subset \R^{N\times N}$ and $\calS_{\mB}$, respectively. Suppose $p\geq q$. Then the principal angles between $\calS_{\mA}$ and $\calS_{\mB}$, $\theta_1(\mA,\mB)\leq \theta_2(\mA,\mB) \leq \cdots \le \theta_q(\mA,\mB)$, are defined as
	\begin{align*}
	\cos\left(\theta_i(\mA,\mB)\right) = \sigma_i(\mA^*\mB)
	\end{align*}
	for all $i\in\left\{1,2,\ldots,q\right\}$, where $\sigma_i(\cdot)$ denotes the $i$-th largest singular value.
\end{definition}
We note that the subspace angle $\Theta_{\mA,\mB}$ is equivalent to the largest principal angle $\theta_q(\mA,\mB)$. To see this, we rewrite the smallest singular value:
\begin{align*}
&\cos\left(\theta_q(\mA,\mB)\right) = \sigma_q(\mA^*\mB) = \inf_{\|\valpha\|_2 =1}\left\|\mA^*\mB \valpha\right\|_2= \inf_{\vb\in\calS_{\mB}\|\vb\|_2 =1,}\left\|\mA^*\vb\right\|_2  = \inf_{\vb\in\calS_{\mB}\|\vb\|_2 =1,}\left\|\mP_{\mA}\vb\right\|_2,
\end{align*}
where the last inequality follows because by assumption $\mA$ is an orthonormal basis for $\calS_{\mA}$. Thus, our definition of subspace angle captures the largest possible principal angle between two subspaces.

Before moving on to prove the main result, we present several results which will also be useful in the remaining proofs. We start with the following result, a variant of Von Neumann's trace inquality~\cite{mirsky1975trace}.

\begin{lem}\label{lem:Von Neumann}\cite{mirsky1975trace}
	For any $M\times N$ (suppose $M\leq N$) matrices $\mA$ and $\mB$ with singular values $\alpha_0\geq \alpha_1 \geq \cdots \geq \alpha_{M-1} \geq 0$ and $\beta_0\geq \beta_1 \geq \cdots \geq \beta_{M-1}\geq 0$, we have
	\[
	\left|\trace(\mA \mB^*) \right| \leq \sum_{m=0}^{M-1} \alpha_m\beta_m.
	\]
\end{lem}
\begin{proof}[Proof of Lemma~\ref{lem:Von Neumann}]
	We enlarge $\mA$ and $\mB$ into $N\times N$ matrices $\mA'$ and $\mB'$ with zero rows, i.e.,
	\[
	\mA' = \left[\begin{array}{c}\mA \\ \mzero  \end{array}\right], \quad \mB' = \left[\begin{array}{c}\mB \\ \mzero  \end{array}\right].
	\]
	Let $\alpha_0'\geq \alpha_1'\geq\cdots \alpha_{N-1}'$ and $\beta_0'\geq \beta_1'\geq\cdots \geq \beta_{N-1}'$ be the singular values of $\mA'$ and $\mB'$, respectively. Note that $\alpha_n = \alpha_n', \beta_n = \beta_n'$ for all $n\leq M-1$ and $\alpha_n'=0, \beta_n'=0$ for all $n\geq M$. It follows from Von Neumann's trace inquality~\cite{mirsky1975trace} that
	\begin{align*}
	&\left|\trace(\mA \mB^*) \right|=\left|\trace(\mA' (\mB')^\H) \right|\leq \sum_{n=0}^{N-1}\alpha_n'\beta_n' = \sum_{m=0}^{M-1} \alpha_m\beta_m.
	\end{align*}
\end{proof}

The following result establishes an upper bound on $\varrho(\mQ)$ in terms of the singular values of $\overline{\mF}_{N,W}^* \mB_{N,W}  - \mV \mV^* \overline{\mF}_{N,W}^* \mB_{N,W}$.
\begin{lem}\label{lem:average SNR bound}
	Let $\mV\in\C^{(N-2\lfloor NW\rfloor -1)\times R}$ be an orthonormal basis with $R\leq (N-2\lfloor NW\rfloor -1)$.
	Let $\pi_0\geq \pi_1 \geq \cdots \geq \pi_{N-2\lfloor NW\rfloor -2}$ denote the singular values of $\overline{\mF}_{N,W}^* \mB_{N,W}  - \mV \mV^* \overline{\mF}_{N,W}^* \mB_{N,W}$. Then
	\[
	\varrho(\mQ) = \int_{-W}^W\left\|\ve_f - \mQ\mQ^*\ve_f\right\|_2^2df \leq \sum_{l=0}^{N-2\lfloor NW\rfloor -2} \pi_l,
	\]
	where $\mQ = \left[\mF_{N,W} \quad \overline{\mF}_{N,W} \mV  \right]$.
\end{lem}
\begin{proof}[Proof of Lemma~\ref{lem:average SNR bound}]
	Recall~\eqref{eq:rho Q to MSE} that
	\begin{align*}
	\varrho(\mQ) = \int_{-W}^W\left\|\ve_f - \mQ\mQ^*\ve_f\right\|_2^2df = \trace\left((\mId - \mQ\mQ^*)\mB_{N,W} \right).
	\end{align*}
	Plugging in $\mQ = \left[\mF_{N,W} \quad \overline{\mF}_{N,W}\mV  \right]$, we have
	\begin{equation}\label{eq:rho Q to trace FBF}\begin{split}
	\varrho(\mQ)&=\trace\left( (\mId - \left[\mF_{N,W} \quad \overline{\mF}_{N,W}\mV  \right]\left[\mF_{N,W} \quad \overline{\mF}_{N,W}\mV  \right]^*)\mB_{N,W}\right)\\
	&= \trace\bigg(\mF_N^*\mB_{N,W}\mF_N - \mF_N^*\begin{bmatrix}\mF_{N,W} & \overline{\mF}_{N,W}\mV  \end{bmatrix}\begin{bmatrix}\mF_{N,W} & \overline{\mF}_{N,W}\mV  \end{bmatrix}^*\mB_{N,W}\mF_N\bigg)\\
	&= \trace\left(\overline{\mF}_{N,W}^* \mB_{N,W}\overline{\mF}_{N,W}  - \mV \mV^* \overline{\mF}_{N,W}^* \mB_{N,W}\overline{\mF}_{N,W}    \right)\\
	&\leq \sum_{l=0}^{N-2\lfloor NW\rfloor -2} \pi_{l}
	\left\|\overline{\mF}_{N,W} \right\| \leq \sum_{l=0}^{N-2\lfloor NW\rfloor -2} \pi_{l},
	\end{split}\end{equation}
	where the first inequality follows from Lemma~\ref{lem:Von Neumann} by setting $\mA = \overline{\mF}_{N,W}^* \mB_{N,W}  - \mV \mV^* \overline{\mF}_{N,W}^* \mB_{N,W}$ and $\mB = \overline{\mF}_{N,W}^*$.
\end{proof}

In order to utilize Lemma~\ref{lem:average SNR bound}, we need the distribution of the singular values of $\overline{\mF}_{N,W}^*\mB_{N,W}$. This is established by the following result, whose proof is given in Appendix~\ref{sec:prf singular value decay of FhighB}.
\begin{lem}\label{lem:singular value decay of FhighB}
	(singular value decay)
	Let $\sigma_0\geq \sigma_1 \geq \cdots \geq \sigma_{N-2\lfloor NW\rfloor -2}$ denote the singular values of $\overline{\mF}_{N,W}^*\mB_{N,W}$. Then
	\[
	\sigma_{\ell} \leq \epsilon
	\]
	when
	$\ell =  C_N \log \left( \frac{15}{\eps} \right) $ for any $\epsilon \in (0,1)$.
	Also
	\[\sigma_{\ell} \leq 15e^{-\frac{\ell}{C_N}}.\]
\end{lem}

Now we are well equipped to prove the main results.
\subsection{Proof of Theorem~\ref{thm:subspace relationship between S_K and Q}}
\begin{proof}[Proof of Theorem~\ref{thm:subspace relationship between S_K and Q}]
	We first provide the following results on the representation guarantee for the leading DPSS vectors and the subspace angle between the column spaces of $\mS_K$ and $\mQ$. The proof of Lemma~\ref{lem:guarantee for representing DPSS} is given in Appendix~\ref{sec:prf guarantee for representing DPSS}.
	\begin{lem}\label{lem:guarantee for representing DPSS}
		Let $\mV\in\C^{(N-2\lfloor NW\rfloor -1)\times R}$ be an orthonormal basis with $R\leq (N-2\lfloor NW\rfloor -1)$. For any $\eps \in (0,\tfrac{1}{2})$, fix $K$ to be such that $\lambda_{N,W}^{(K-1)}\geq\epsilon$. Let
		\[
		\eta:=\frac{  \left\| \overline{\mF}_{N,W}^*\mB_{N,W}- \mV\mV^* \overline{\mF}_{N,W}^* \mB_{N,W} \right\| }{\epsilon}
		\]
		Then
		the orthonormal basis $\mQ = \left[\mF_{N,W} \quad \overline{\mF}_{N,W} \mV  \right]$ satisfies
		\begin{align*}
		&\|\mS_K\mS_K^* - \mQ\mQ^*\mS_K \mS_K^*\|^2 \leq \eta,\\
		& \cos\left(\Theta_{\mS_K,\mQ}\right)\geq \sqrt{1 - N\eta}, \\
		&\left\| \vs_{N,W}^{(\ell)} - \mQ \mQ^*\vs_{N,W}^{(\ell)}  \right\|_2^2\leq \eta
		\end{align*}
		for all $l = 0,1,\ldots, K-1.$
	\end{lem}
	
	Since $\mV$ contains the first $R$ principal eigenvectors of $\overline{\mF}_{N,W}^*\mB_{N,W}$, using Lemma~\ref{lem:singular value decay of FhighB}, we obtain
	\[
	\left\| \overline{\mF}_{N,W}^*\mB_{N,W}- \mV\mV^* \overline{\mF}_{N,W}^* \mB_{N,W} \right\| \leq 15e^{-\frac{R}{C_N}}.
	\]
	If we set $R = C_N \log \left( \frac{15}{\epsilon^2} \right)$, we have
	\[\left\| \overline{\mF}_{N,W}^*\mB_{N,W}- \mV\mV^* \overline{\mF}_{N,W}^* \mB_{N,W} \right\|\leq \eps^2.\]
	Alternatively, if one set $R = C_N \log \left( \frac{15N}{\epsilon^2} \right)$:
	\[\left\| \overline{\mF}_{N,W}^*\mB_{N,W}- \mV\mV^* \overline{\mF}_{N,W}^* \mB_{N,W} \right\| \leq \frac{\epsilon^2}{N}.\]
	
	The proof of Theorem~\ref{thm:subspace relationship between S_K and Q} completed by utilizing Lemma~\ref{lem:guarantee for representing DPSS}.
\end{proof}

\subsection{Proof of Theorem~\ref{thm:average repres error}}\label{sec: proof average repres error}
\begin{proof}[Proof of Theorem~\ref{thm:average repres error}]
	Let $\sigma_0\geq \sigma_1 \geq \cdots \geq \sigma_{N-2\lfloor NW\rfloor -2}$ denote the singular values of $\overline{\mF}_{N,W}^*\mB_{N,W}$. Since $\mV$ consists of the $R$ dominant left singular vectors of $\overline{\mF}_{N,W}^* \mB_{N,W}$, the singular values of
	$\overline{\mF}_{N,W}^* \mB_{N,W} - \mV \mV^* \overline{\mF}_{N,W}^* \mB_{N,W}$ are $\sigma_{R},\sigma_{R+1},\ldots, \sigma_{N - 2\lfloor NW\rfloor -2}$ and $R$ zeros.
	It follows from Lemma~\ref{lem:singular value decay of FhighB} that
	\begin{equation}\label{eq:sum of sig values}
	\begin{split}
	\sum_{\ell = R}^{N-2\lfloor NW\rfloor -2} \sigma_\ell \leq & \sum_{\ell = R}^{N-2\lfloor NW\rfloor -2} 15e^{-\frac{\ell}{C_N}}\\
	= & 15 \frac{e^{-\frac{R}{C_N}} (1- e^{-\frac{N-2\lfloor NW\rfloor  -R-1}{C_N}})}{1-e^{-\frac{1}{C_N}}}\\
	\leq & 15 \frac{e^{-\frac{R}{C_N}}}{1-e^{-\frac{1}{C_N}}} = 15 \frac{e^{-\frac{R-1}{C_N}}}{e^{\frac{1}{C_N}}-1}\\
	\leq & 15 e^{-\frac{R-1}{C_N}} C_N,
	\end{split}\end{equation}
	where the last line  holds because $e^{a-1}\geq a$ for all $a\geq 0$.
	
	If $C_N \log \left( \frac{15 C_N}{N\eps} \right) +1\leq 0$, which implies that
	\[
	\sum_{\ell = 0}^{N-2\lfloor NW\rfloor -2} \sigma_\ell \leq N\epsilon,
	\]
	then by setting $R=0$ and $\mQ = \mF_{N,W}$ we are guaranteed that
	\[
	\int_{-W}^W\frac{\left\|\ve_f - \mQ\mQ^*\ve_f\right\|_2^2}{\|\ve_f\|_2^2}df \le \frac{1}{N} \sum_{\ell=0}^{N-2\lfloor NW\rfloor -2} \sigma_{\ell} \le \frac{1}{N}N\eps= \epsilon.
	\]
	Otherwise, choosing $R = C_N \log \left( \frac{15 C_N}{N\eps} \right) +1$, we have
	\begin{align*}
	\sum_{\ell = R}^{N-2\lfloor NW\rfloor -2} \sigma_\ell \leq & N\eps.
	\end{align*}
	Now applying Lemma~\ref{lem:average SNR bound}, we have
	\[
	\int_{-W}^W\frac{\left\|\ve_f - \mQ\mQ^*\ve_f\right\|_2^2}{\|\ve_f\|_2^2}df \le \frac{1}{N} \sum_{\ell=R}^{N-2\lfloor NW\rfloor -2} \sigma_{\ell} \le \frac{1}{N}N\eps= \epsilon,
	\]
	where we utilize the fact that each sinusoid has energy $\|\ve_f\|_2^2 = N$. This completes the proof of Theorem~\ref{thm:average repres error}.
\end{proof}

\begin{remark}By~\eqref{eq:rho Q to trace FBF}, we have
	\begin{align*}
	\varrho(\mQ) = \trace\left(\overline{\mF}_{N,W}^* \mB_{N,W}\overline{\mF}_{N,W}  - \mV \mV^* \overline{\mF}_{N,W}^* \mB_{N,W}\overline{\mF}_{N,W}    \right).
	\end{align*}
	Directly solving
	\begin{align*}
	\minimize_{\mV \in \C^{(N-2\lfloor NW\rfloor -1)\times R} }\trace\big(&\overline{\mF}_{N,W}^* \mB_{N,W}\overline{\mF}_{N,W}  - \mV \mV^* \overline{\mF}_{N,W}^* \mB_{N,W}\overline{\mF}_{N,W}    \big),
	\end{align*}
	we obtain an alternative optimal solution $\widetilde\mV$ consisting of the first $R$ principal eigenvectors of $\overline{\mF}_{N,W}^* \mB_{N,W}\overline{\mF}_{N,W}$. The orthonormal basis $\mQ' = \left[\mF_{N,W} \quad \overline{\mF}_{N,W} \widetilde\mV  \right]$  is optimal in terms of minimizing $\varrho(\mQ)$ and also for representing all discrete-time sinusoids with a frequency $f\in[-W,W]$ in the least square sense. Similar to Theorem \ref{thm:average repres error}, we can also establish an approximation guarantee for $\widetilde\mV$. Note that
	\begin{align*}
	&\mF_N^*\left(\mB_{N,W} - \mF_{N,W}\mF_{N,W}^*\right) \mF_N= \left[\begin{array}{cc}\mF_{N,W}^*\mB_{N,W}\mF_{N,W} - \mId &  \mF_{N,W}^*\mB_{N,W}\overline{\mF}_{N,W}\\
	\overline{\mF}_{N,W}^*\mB_{N,W}\mF_{N,W} &  \overline{\mF}_{N,W}^*\mB_{N,W}\overline{\mF}_{N,W}
	\end{array}\right].
	\end{align*}
	By utilizing the result that \[
	\mB_{N,W} = \mF_{N,W} \mF_{N,W}^* + \mL + \mE,
	\]
	where
	\[
	\rank(\mL) \le C_N \log \left( \frac{15}{\eps} \right) \quad \quad \text{and} \quad \quad  \| \mE \| \le \eps,
	\]
	we can rewrite $\overline{\mF}_{N,W}^*\mB_{N,W}\overline{\mF}_{N,W} = \mL_2 + \mE_2$, where
	\[
	\mL_2: =\overline{\mF}_{N,W}^*\mL \overline{\mF}_{N,W} \quad \quad \text{and} \quad\quad  \mE_2: =\overline{\mF}_{N,W}^*\mE \overline{\mF}_{N,W}.
	\]
	Thus,
	\[
	\rank(\mL_2) \le C_N \log \left( \frac{15}{\eps} \right) \quad \quad \text{and} \quad \quad  \| \mE_2 \| \le \eps.
	\]
	It follows from the Eckart-Young-Mirsky theorem~\cite{eckart1936approximation} that
	\[
	\|\overline{\mF}_{N,W}^*\mB_{N,W}\overline{\mF}_{N,W} - \widetilde\mV\widetilde\mV^* \overline{\mF}_{N,W}^* \mB_{N,W}\overline{\mF}_{N,W}\|\leq \|\mE_2\| \leq \epsilon.
	\]
	Therefore, choosing $R = C_N \log \left( \frac{15 C_N}{N\eps} \right) +1$, with a similar argument we also have
	\begin{align*}
	&\int_{-W}^W\frac{\left\|\ve_f - \mQ\mQ^*\ve_f\right\|_2^2}{\|\ve_f\|_2^2}df \le \frac{1}{N} \trace\left(\overline{\mF}_{N,W}^* \mB_{N,W}\overline{\mF}_{N,W}  - \widetilde\mV\widetilde\mV^* \overline{\mF}_{N,W}^* \mB_{N,W}\overline{\mF}_{N,W}\right) \le \epsilon.
	\end{align*}
	We note that all other results in this paper involving $\mV$ can also be applied to $\widetilde\mV$ with similar or slightly different guarantees.
\end{remark}
\subsection{Proof of Theorem \ref{thm:fast for sinusoid}}
By Theorem~\ref{thm:average repres error}, we are guaranteed that the pure sinusoids have, on average, a small representation residual in the basis $\mQ$. Intuitively, the representation error for each pure sinusoid is also guaranteed to be small. The following result provides an upper bound on the representation error for each pure sinusoid in terms of the average representation error. Its proof is given in Appendix~\ref{sec:prf uniform guarantee to pure sinusoid}.
\begin{lem}\label{lem:uniform guarantee to pure sinusoid}
	For any $q\in\{1,2,\ldots,N\}$, suppose $\mU\in\C^{N\times q}$ is an orthonormal basis such that $\mU^*\mU = \mId$. Also suppose $W\geq \frac{1}{4\pi N}$. Then
	\begin{align*}
	\frac{\left\|\ve_f - \mU\mU^*\ve_f\right\|_2^2}{\|\ve_f\|_2^2} \leq \max\bigg(&2\sqrt{\pi}\sqrt{ \int_{-W}^W\left\|\ve_f - \mU\mU^*\ve_f\right\|_2^2df}, \frac{1}{NW}\int_{-W}^W\left\|\ve_f - \mU\mU^*\ve_f\right\|_2^2 \bigg).
	\end{align*}
\end{lem}

\begin{proof}[Proof of Theorem~\ref{thm:fast for sinusoid}]
	It follows from \eqref{eq:sum of sig values} that by choosing $R = C_N \log \left( \frac{15 C_N}{\eps'} \right) +1$, we have
	\begin{align*}
	\sum_{l = R}^{N-2\lfloor NW\rfloor -1} \sigma_l \leq & \eps'.
	\end{align*}
	Utilizing Lemma~\ref{lem:average SNR bound} gives
	\[
	\int_{-W}^W\left\|\ve_f - \mQ\mQ^*\ve_f\right\|_2^2df \leq \sum_{l = R}^{N-2\lfloor NW\rfloor -1} \sigma_l\leq \epsilon'.
	\]
	The proof of Theorem~\ref{thm:fast for sinusoid} is completed by setting
	\[\eps' = \frac{\eps^2}{4\pi}, \quad R = C_N \log \left( \frac{60\pi C_N}{\eps^2} \right) +1,\]
	or
	\[\eps' = NW\eps, \quad R = C_N \log \left( \frac{15 C_N}{NW\eps} \right) +1.\]

\end{proof}

\subsection{Proof of Theorem~\ref{thm:randomized algorithm for finding V}}
We first present the following guarantees on randomized algorithms for computing orthonormal bases from~\cite{halkoRandomizedAlgorithm}.
\begin{thm}\label{thm:average Frobenius norm}
	\cite[Theorem 10.5]{halkoRandomizedAlgorithm} (Average Frobenius norm)
	Let $\mA$ be an $M\times N$ (suppose $M\leq N$) matrix with singular values $\alpha_0 \geq \alpha_1 \geq \cdots \ \alpha_{M-1}$. Choose a target rank $R\geq 2$ and an oversampling parameter $p\geq2$, where $P = R+p \leq M$. Let $\mOmega$ be an $N\times P$ standard Gaussian matrix. Let $\mP_{\mY}$ be an orthogonal projector onto the column space of the sample matrix $\mY = \mA \mOmega$. Then the expected approximation error
	\[
	\mathbb{E}\left[\left\|\mA - \mP_{\mY} \mA \right\|_F\right] \leq \left(1 + \frac{R}{p-1}  \right)^{1/2}\left(\sum_{m=R}^{M-1} \alpha_m^2 \right)^{1/2},
	\]
	where $\mathbb{E}$ denotes expectation with respect to the random matrix $\mOmega$.
\end{thm}

\begin{thm}\label{thm:average spectral error}\cite[Theorem10.6]{halkoRandomizedAlgorithm} (Average spectral error)
	Under the setup of Theorem~\ref{thm:average Frobenius norm},
	\[
	\mathbb{E}\left[\left\|\mA - \mP_{\mY} \mA \right\|\right] \leq \left(1 + \sqrt{\frac{R}{p-1}} \right)\alpha_{R} + \frac{e\sqrt{P}}{p}\left(\sum_{m=R}^{M-1} \alpha_m^2 \right)^{1/2}.
	\]
\end{thm}
\begin{proof}[Proof of Theorem~\ref{thm:randomized algorithm for finding V}]
	Let $\sigma_0\geq \sigma_1 \geq \cdots \geq \sigma_{N-2\lfloor NW\rfloor -2}$ denote the singular values of $\overline{\mF}_{N,W}^*\mB_{N,W}$. Utilizing Lemma~\ref{lem:singular value decay of FhighB}, we have
	\begin{align*}
	\sum_{l = R}^{N-2\lfloor NW\rfloor -2} \sigma_l^2 &\leq  \sum_{l = R}^{N-2\lfloor NW\rfloor -2} \left(15e^{-\frac{\ell}{C_N}}\right)^2\\
	& =  225 \frac{e^{-2\frac{R}{C_N}} (1- e^{-2\frac{N-2\lfloor NW\rfloor  -R-1}{C_N}})}{1-e^{-\frac{2}{C_N}}}\\
	& \leq 225 \frac{e^{-2\frac{R}{C_N}}}{1-e^{-\frac{2}{C_N}}} = 225 \frac{e^{-2\frac{R-1}{C_N}}}{e^{\frac{2}{C_N}}-1}\\
	& \leq  225 e^{-2\frac{R-1}{C_N}} \frac{C_N}{2}.
	\end{align*}
	Note that here $\overline\mV$ is an orthonormal basis for the column space of the sample matrix $\overline{\mF}_{N,W}^*\mB_{N,W}\mOmega$. Let $\pi_0\geq \pi_1 \geq \cdots \geq \pi_{N-2\lfloor NW\rfloor -2}$ denote the singular values of $\overline{\mF}_{N,W}^* \mB_{N,W}  - \overline\mV \overline\mV^* \overline{\mF}_{N,W}^* \mB_{N,W}$.
	
	Show $(i)$:
	Utilizing Theorem~\ref{thm:average spectral error}, we have
	\begin{align*}
	&\mathbb{E}\left[\left\|\overline{\mF}_{N,W}^* \mB_{N,W} - \overline\mV \overline\mV^* \overline{\mF}_{N,W}^* \mB_{N,W}\right\|\right]\\
	& \leq \left(1 + \sqrt{\frac{R}{P-R-1}} \right)\sigma_{R} + \frac{e\sqrt{P}}{P-R}\left(\sum_{l=R}^{M-1} \sigma_l^2 \right)^{1/2}\\
	&\leq \left(1 + \sqrt{\frac{R}{P-R-1}} \right)15e^{-\frac{R}{C_N}}+ \frac{e\sqrt{P}}{P-R}\left(225 e^{-2\frac{R-1}{C_N}} \frac{C_N}{2} \right)^{1/2}\\
	&=  \left(1 + \sqrt{\frac{R}{P-R-1}} \right)15e^{-\frac{R}{C_N}} + 15\frac{e\sqrt{P}}{P-R} e^{-\frac{R-1}{C_N}} \sqrt{\frac{C_N}{2}}.
	\end{align*}
	Setting $R = C_N \log \left( \frac{30+15e}{\eps^2} \right)+1$ and $P = 2R+1$, we have
	\begin{align*}
	&\mathbb{E}\left[\left\|\overline{\mF}_{N,W}^* \mB_{N,W} - \overline\mV \overline\mV^* \overline{\mF}_{N,W}^* \mB_{N,W}\right\|\right]\leq 30\frac{\epsilon^2}{30+15e} + 15e\sqrt{\frac{C_N}{R+1}}\frac{\epsilon^2}{30+15e}\leq \epsilon^2
	\end{align*}
	since $C_N\leq R$ for any $\epsilon^2 \in(0,1)$. It follows from Lemma~\ref{lem:guarantee for representing DPSS} that
	\begin{align*}
	&\mathbb{E} \left[\|\mS_K\mS_K^* - \mQ\mQ^*\mS_K \mS_K^*\|^2\right] \leq \frac{\mathbb{E}\left[\left\|\overline{\mF}_{N,W}^* \mB_{N,W} - \overline\mV \overline\mV^* \overline{\mF}_{N,W}^* \mB_{N,W}\right\|\right]}{\epsilon}\leq \epsilon,\ \mathbb{E} \left[\| \vs_{N,W}^{(\ell)} - \mQ \mQ^*\vs_{N,W}^{(\ell)}  \|^2\right]\leq  \epsilon
	\end{align*}
	for all $l = 0,1,\ldots, K-1$. Alternatively, setting $R = C_N \log \left( \frac{\left(30+15e\right)N}{\eps^2} \right)+1$ and $P = 2R+1$, we have
	\begin{align*}
	\mathbb{E}\left[\left\|\overline{\mF}_{N,W}^* \mB_{N,W} - \overline\mV \overline\mV^* \overline{\mF}_{N,W}^* \mB_{N,W}\right\|\right] \leq \frac{\epsilon^2}{N}.
	\end{align*}
	Thus applying Lemma~\ref{lem:guarantee for representing DPSS} gives
	\begin{align*}
	&\mathbb{E}\left[\cos\left(\Theta_{\mS_K,\mQ}\right)\right]\geq \sqrt{1 - N\frac{\mathbb{E}\left[\left\|\overline{\mF}_{N,W}^* \mB_{N,W} - \overline\mV \overline\mV^* \overline{\mF}_{N,W}^* \mB_{N,W}\right\|\right]}{\epsilon}}\geq \sqrt{1-\eps}.
	\end{align*}
	
	Show $(ii)$: Set $p = \frac{R}{3} + 1$, i.e., $P = \frac{4}{3}R+1$. It follows from Theorem~\ref{thm:average Frobenius norm} that
	\begin{align*}
	\mathbb{E}\left\|\overline{\mF}_{N,W}^* \mB_{N,W} - \overline\mV \overline\mV^* \overline{\mF}_{N,W}^* \mB_{N,W}\right\|_F& \leq  \left(1 + \frac{R}{p-1}  \right)^{1/2}\left(\sum_{l=R}^{N-2\lfloor NW\rfloor-2} \sigma_l^2 \right)^{1/2}\\
	&\leq 2\sqrt{225 e^{-2\frac{R-1}{C_N}} \frac{C_N}{2}}=  15 e^{-\frac{R-1}{C_N}} \sqrt{2C_N}.
	\end{align*}
	By applying Lemma~\ref{lem:average SNR bound} and utilizing the inequality between the Frobenius norm and the nuclear norm, we have
	\begin{equation}\label{eq:average snr for randomized algorithm}\begin{split}
	&\mathbb{E}\int_{-W}^W\frac{\left\|\ve_f - \mQ\mQ^*\ve_f\right\|_2^2}{\|\ve_f\|_2^2}df = \frac{1}{N}\mathbb{E} \sum_{m=0}^{N-2\lfloor NW\rfloor-2}\pi_m \\
	&\leq \frac{1}{N} N\mathbb{E}\left\|\overline{\mF}_{N,W}^* \mB_{N,W} - \overline\mV \overline\mV^* \overline{\mF}_{N,W}^* \mB_{N,W}\right\|_F\\ &\leq 15 e^{-\frac{R-1}{C_N}} \sqrt{2C_N}.
	\end{split}\end{equation}
	Setting $R = C_N \log \left( \frac{15\sqrt{2C_N}}{\eps} \right)+1$, we obtain
	\[
	\mathbb{E}\int_{-W}^W\frac{\left\|\ve_f - \mQ\mQ^*\ve_f\right\|_2^2}{\|\ve_f\|_2^2}df \leq \epsilon.
	\]

	Show $(iii)$: Set $p = \frac{R}{3} + 1$, i.e., $P = \frac{4}{3}R+1$. From~\eqref{eq:average snr for randomized algorithm}, it follows that
	\[
	\mathbb{E}\left[\int_{-W}^W\left\|\ve_f - \mQ\mQ^*\ve_f\right\|_2^2\right]
	\le 15N e^{-\frac{R-1}{C_N}} \sqrt{2C_N}.
	\]
	Utilizing Lemma~\ref{lem:uniform guarantee to pure sinusoid}, we have
	\begin{align*}
	\mathbb{E}\left[\frac{\left\|\ve_f - \mQ\mQ^*\ve_f\right\|_2^2}{\|\ve_f\|_2^2}\right]
	&\leq  \max\bigg(\mathbb{E}\bigg[ 2\sqrt{\pi}\sqrt{ \int_{-W}^W\left\|\ve_f - \mQ\mQ^*\ve_f\right\|_2^2df}\bigg],\mathbb{E}\bigg[ NW\int_{-W}^W\left\|\ve_f - \mQ\mQ^*\ve_f\right\|_2^2\bigg]\bigg)\\
	&\leq  \max\left(2\sqrt{15\pi N \sqrt{2C_N}}e^{-\frac{R-1}{2C_N}}, \frac{15N e^{-\frac{R-1}{C_N}} \sqrt{2C_N}}{NW} \right).
	\end{align*}
	Setting
	\begin{align*}
	R = \max\bigg(& C_N \log \left( \frac{
		60\pi N\sqrt{2C_N }}{\eps^2} \right)+1,C_N \log \left( \frac{
		15\pi \sqrt{2C_N }}{W\eps} \right)+1 \bigg)
	\end{align*}
	yields
	\[
	\mathbb{E}\left[\frac{\left\|\ve_f - \mQ\mQ^*\ve_f\right\|_2^2}{\|\ve_f\|_2^2}\right] \leq \epsilon.
	\]
\end{proof}

\section{Simulations}\label{sec:simulations}
In this section, we present some experiments to illustrate the effectiveness of our proposed ROAST and ROAST-R (which is short for ROAST with a Randomized algorithm for computing $\mV$---see Section~\ref{sec:roast2}). Throughout this section, we use $R$ (which is typically $O(\log(N))$) to denote the the dimensionality of $\mV$ for ROAST. For ROAST-R, we set $P$, the dimensionality of $\overline\mV$, as $P=R$ here.

For comparison, we also compute the projection onto the column space of $\mF_{N,W+\frac{R}{2N}}$ which is the $N\times (2\lfloor  NW\rfloor+1 + R)$ DFT matrix with frequencies in $[-W-\frac{R}{2N},W+\frac{R}{2N}]$. Such a projection is simply denoted by Sub-DFT.  Note that the dimension of the column space of $\mF_{N,W+\frac{R}{2N}}$ is $2\lfloor  NW\rfloor+1 + R$ and is equal to the dimension of the column space of ${\mQ}$  in ROAST and ROAST-R. We also compare with DPSS since it provides the gold standard in approximation performance. Specifically, the projection onto the column space of the leading DPSS vectors $\mS_K$ is computed and denoted simply by DPSS in the legends of the figures. We also choose $K = 2\lfloor  NW\rfloor+1 + R$ so that all these subspaces have the same dimensionality.

We quantify the ability of the different projections to capture a given signal $\vx\in\C^N$ in terms of
$$\text{SNR} = 20\log_{10}\left(\frac{\|\vx\|_2}{\|\vx - \widehat {\vx}\|_2}\right)\text{dB},$$
where $\widehat \vx$ is the resulting projection of $\vx$ by the above mentioned methods.

Figure~\ref{fig:SNR Pure Sinusoid and Bandlimited}(a) shows the SNR captured by different projections for various pure sinusoids $\ve_f$. We observe that the DPSS basis, ROAST, ROAST-R and provide almost equal approximation performance for the pure sinusoids with frequencies in the band of interest. Also as guaranteed by Theorems~\ref{thm:fast for sinusoid}, \ref{thm:randomized algorithm for finding V} and \cite[Theorem 3.9]{ZhuWakin2015MDPSS}, any sinusoid in the band of interest can be well represented by the DPSS basis, ROAST, and ROAST-R.

We also generate a sampled bandlimited signal $\vx$ by adding $10^5$ complex exponentials with frequencies selected uniformly at random within the frequency band $[-W,W]$.
Figure~\ref{fig:SNR Pure Sinusoid and Bandlimited}(b) shows the ability of the different projections to capture this vector in terms of $\text{SNR}$.  Again, it can be observed that the DPSS basis, ROAST, and ROAST-R provide almost equal approximation performance for sampled bandlimited signals.

We now compare ROAST with FST (see \eqref{eq:FST}) which involves two skinny matrices $\mD_1,\mD_2\in\R^{N\times r}$ with
\begin{align*}
r \leq 3 C_N \log \left( \frac{15}{\delta} \right),
\end{align*}
where $\delta$ is the approximation accuracy and is chosen as $\delta = 10^{-5}$ unless stated otherwise. As we explained in Section~\ref{sec:comparision}, in some applications $r$ is prescribed instead of the approximation accuracy $\delta$. For these cases, we modify the FST such that $\mD_1$ and $\mD_2$ have the same number of columns as $\mV$ (i.e., $r = R$). The corresponding transform is denoted by  FST-FR (shorted for FST with Fixed Rank)\footnote{We note that the code for this transform is not optimized. For FST-FR, we set $K = 2\lfloor  NW\rfloor+1 + \lfloor\frac{1}{4}R\rfloor$ as we require that $\mD_1$ and $\mD_2$ have the same number of columns as $\mV$.}.

We compare the size, speed, and approximation performance the six projection methods. In these experiments, we fix $R = \lfloor 3\log(N)\rfloor$ and $\delta = 10^{-5}$.
Figures~\ref{fig:vary N}(a) and (b) respectively plot SNR as a function of dimension $N$ and  the relationship between the run time and $N$ for the six projection methods.  As observed, the DPSS has the best approximation performance as guaranteed by \eqref{eq:mse with trace} and \eqref{eq:redidual by DPSS}, but
the running time of DPSS has a quadratic increase. FST, FST-FR, ROAST and ROAST-R\footnote{FST-FR, ROAST, and ROAST-R are expected to have the same running time since these three transforms have the same dimensionality and form.} are nearly as fast as the DFT, but with much better approximation performance (except FST-FR which only has slightly better approximation quality than the DFT). Figure~\ref{fig:vary N}(c) shows the precomputation time needed for the five projection methods. For the DPSS basis, the first $K$ DPSS vectors are precomputed with the Matlab command {\tt dpss} (which actually computes the eigenvectors of a tridiagonal matrix with computational complexity of $O(N^2)$).  As can be seen in Figure~\ref{fig:vary N}(c), the precomputation time required by the DPSS grows roughly quadratically with $N$, while the precomputation time required by other fast transforms grows just faster than linearly in $N$. Figure~\ref{fig:vary N}(d) compares the value of $r$ (the number of columns of $\mD_1$ and $\mD_2$ for FST) and $R$ (the number of columns of the skinny matrices in ROAST, ROAST-R, and FST-FR). In a nutshell, we see that FST has a similar approximation quality, but at the expense of a larger and slower transform. On the other hand, when we fix the size of FST the same as ROAST and ROAST-R, as depicted in Figure~\ref{fig:vary N}(a), FST-FR has inferior approximation quality to ROAST and ROAST-R.

\begin{figure*}[htb!]
	\begin{minipage}{0.48\linewidth}
		\centerline{
			\includegraphics[width=2.9in]{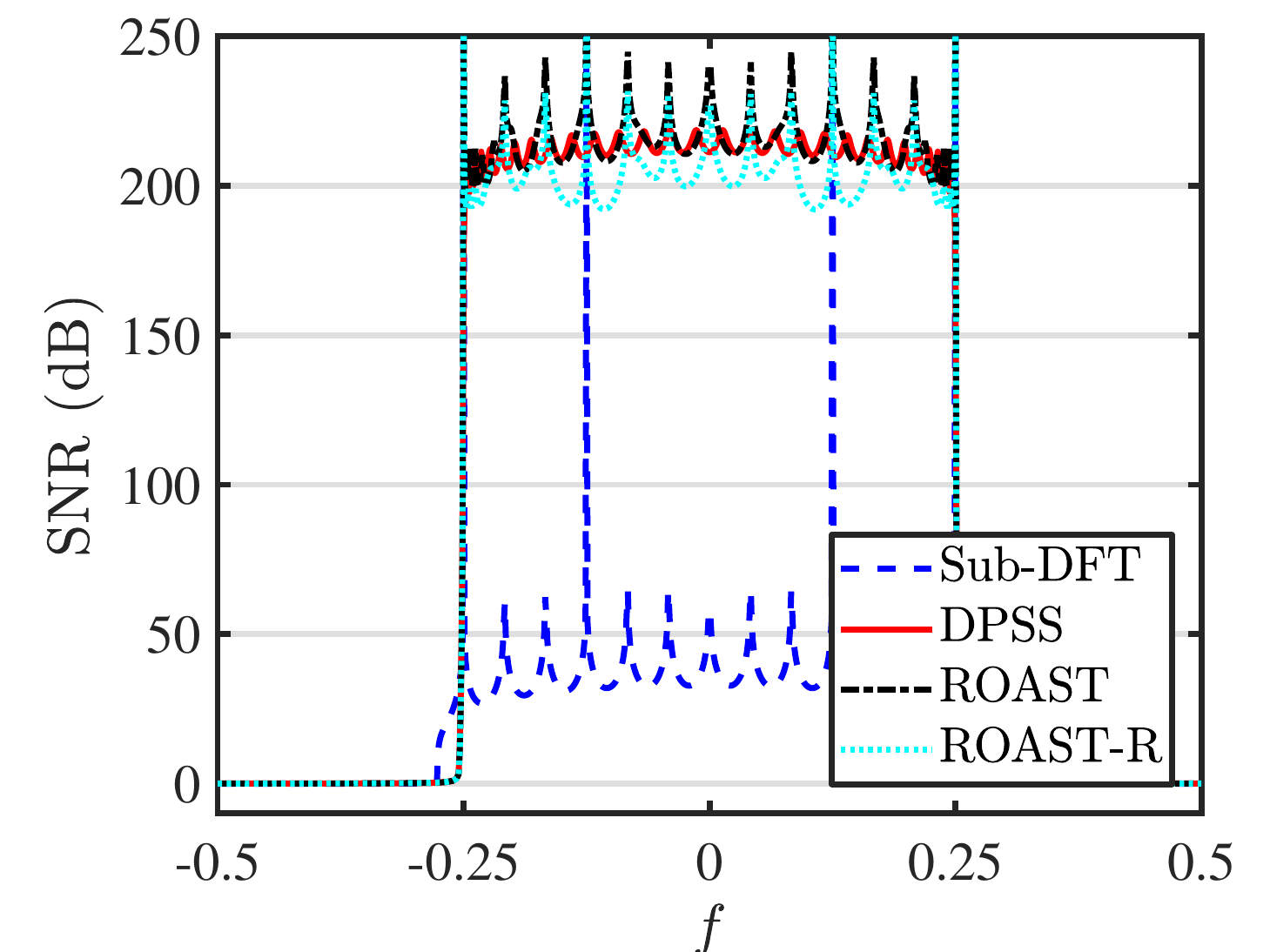}}
		\centering{(a)}
	\end{minipage}
	\hfill
	\begin{minipage}{0.48\linewidth}
		\centerline{
			\includegraphics[width=2.9in]{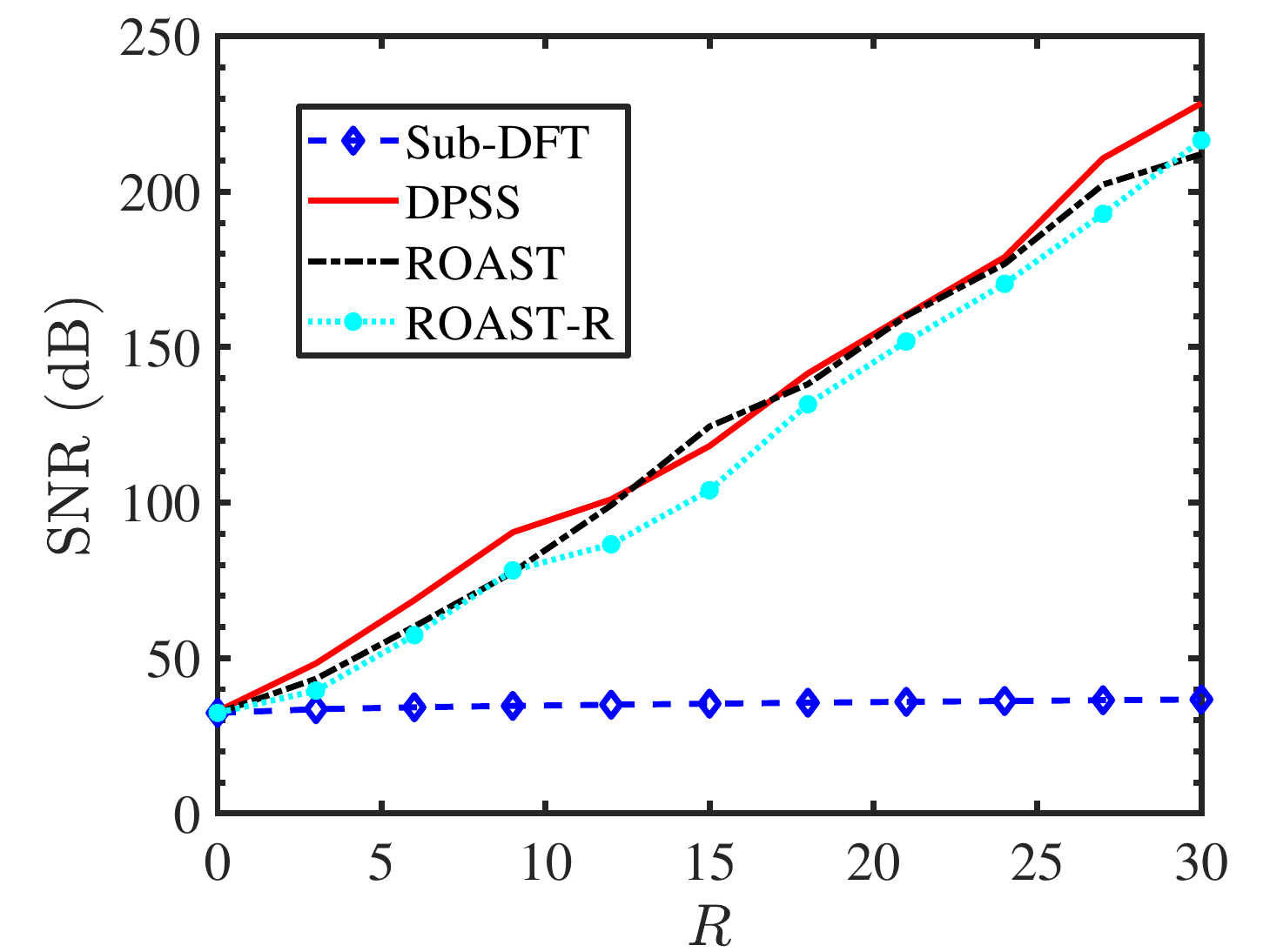}}
		\centering{(b)}
	\end{minipage}
	\caption{ (a) SNR captured by different projections  for pure sinusoids $\ve_f$ with $R = 4\log(N)$; (b) SNR captured by different projections for a sampled bandlimited signal $\vx$ with $R$ ranging from $0$ to $30\approx 5\log(N)$. Here $N = 1024$, $W = \frac{1}{4}$. }\label{fig:SNR Pure Sinusoid and Bandlimited}
\end{figure*}

%

\begin{figure*}[htb!]
	\begin{minipage}{0.48\linewidth}
		\centerline{
			\includegraphics[width=2.9in]{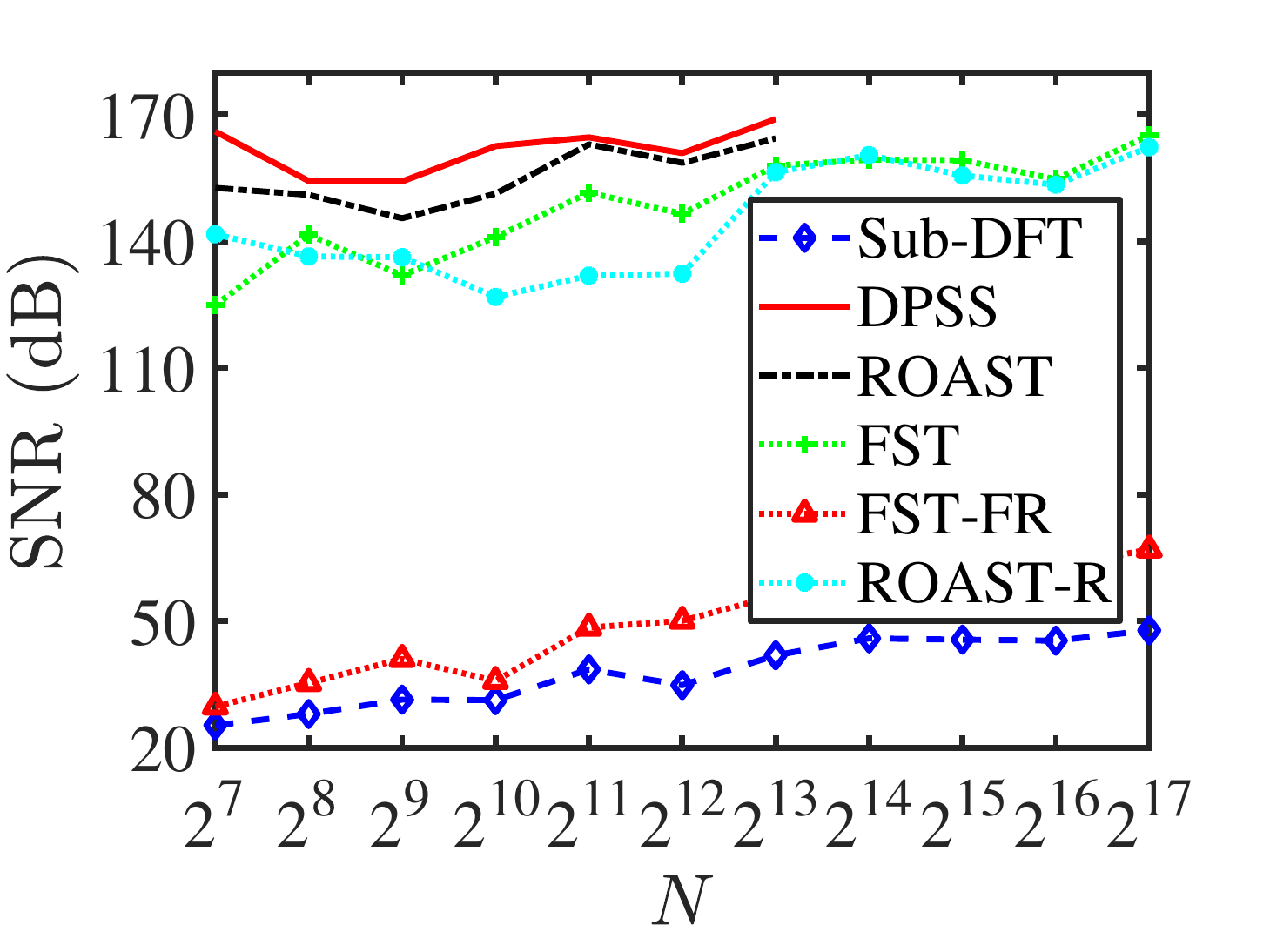}}
		\centering{(a)}
	\end{minipage}
	\hfill
	\begin{minipage}{0.48\linewidth}
		\centerline{
			\includegraphics[width=2.9in]{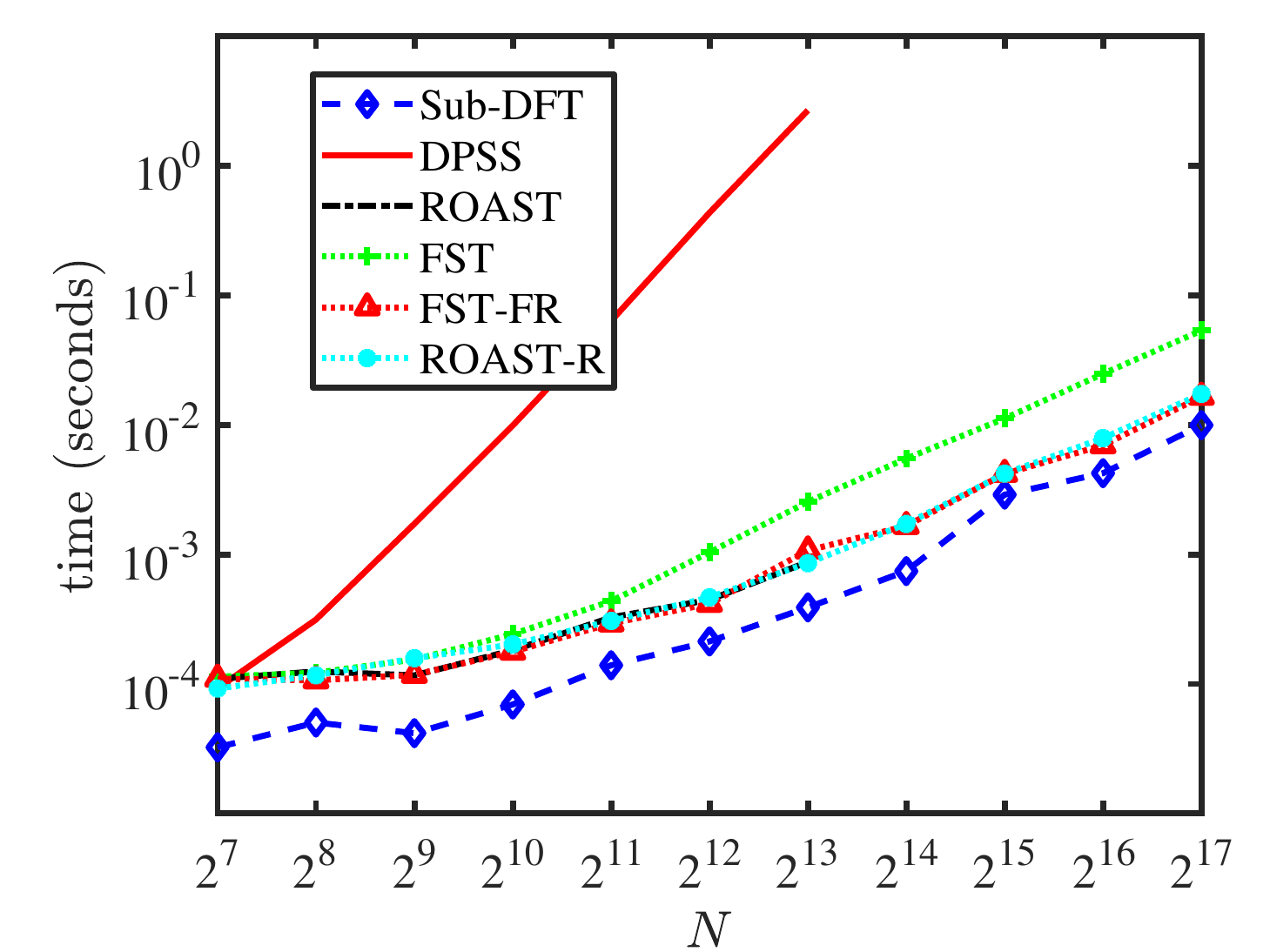}}
		\centering{(b)}
	\end{minipage}
	\begin{minipage}{0.48\linewidth}
		\centerline{
			\includegraphics[width=2.9in]{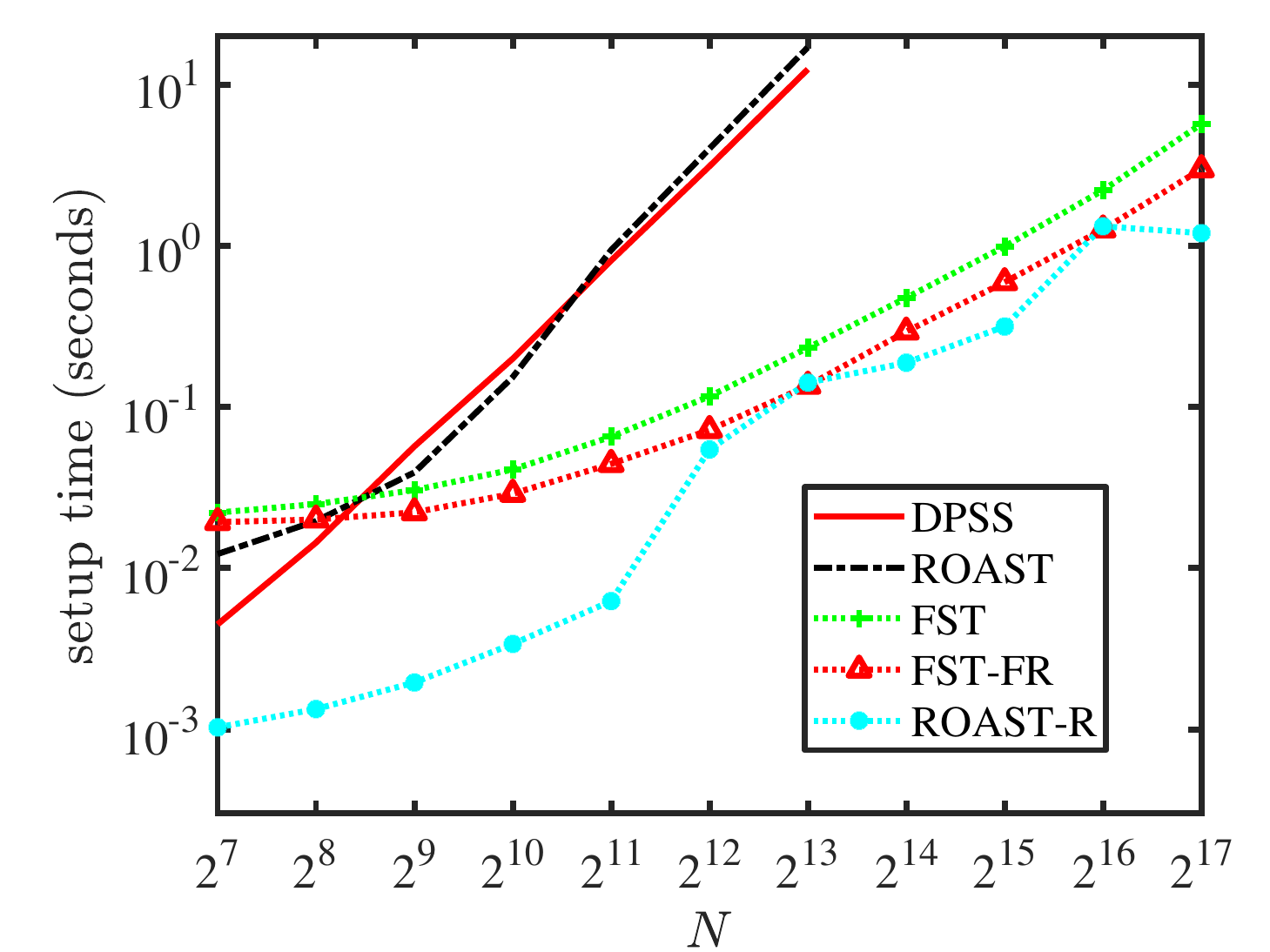}}
		\centering{(c)}
	\end{minipage}
	\hfill
	\begin{minipage}{0.48\linewidth}
		\centerline{
			\includegraphics[width=2.9in]{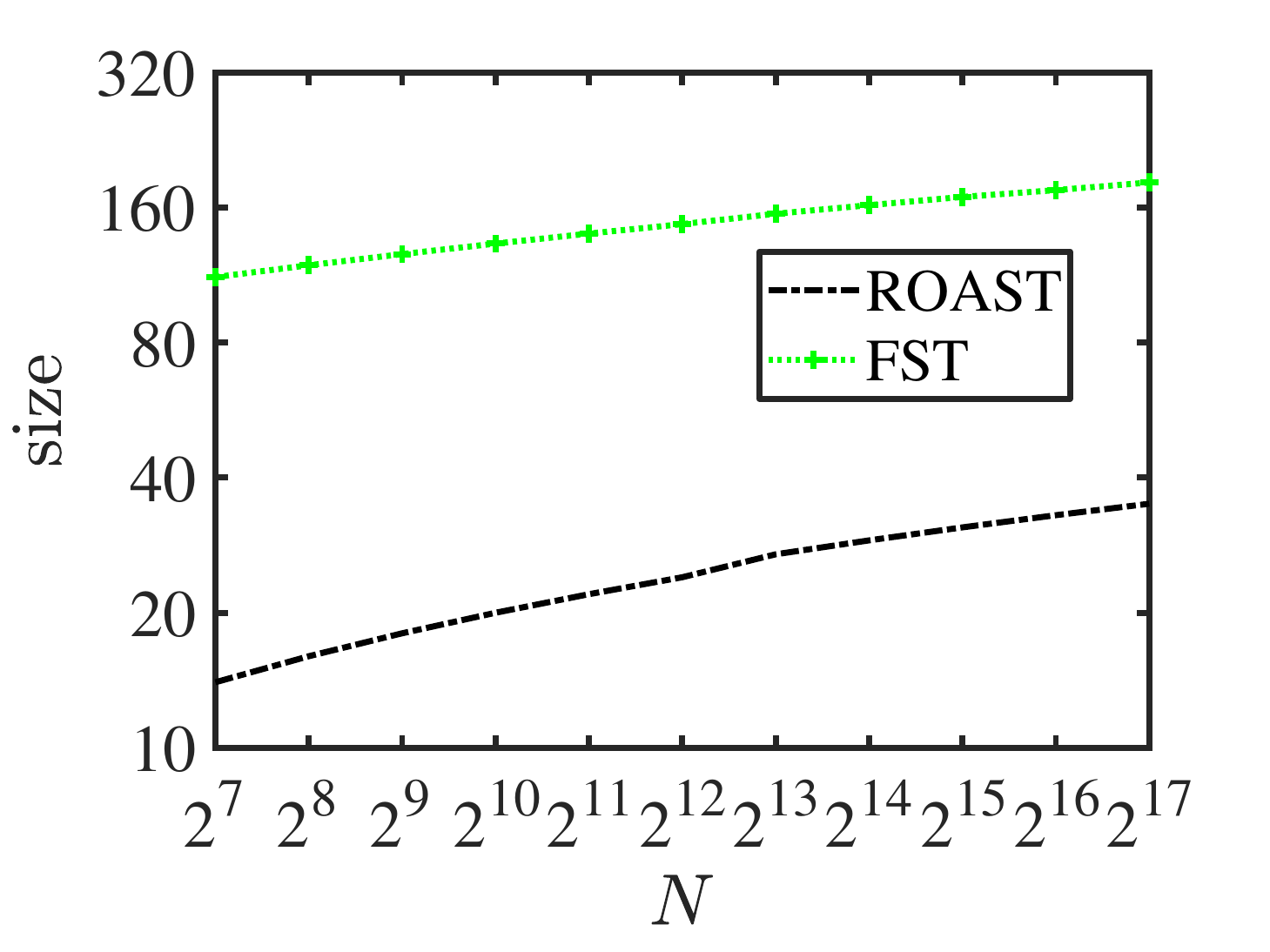}}
		\centering{(d)}
	\end{minipage}
	\caption{Comparison of different projections for a sampled bandlimited signal $\vx$: (a) SNR as a function of $N$;  (b) computation time as a function of $N$ in a logarithmic scale for both axes (here the lines for ROAST, ROAST-R and FST-FR overlapped together since these three transforms have the same dimension); (c) precomputation time; (d) number of columns of the skinny matrices (i.e., $R$ for ROAST and $r$ for FST). In all plots, $W = \frac{1}{4}$, $R = \lfloor 3\log(N)\rfloor$ for Sub-DFT, DPSS, ROAST, ROAST-R and FST, and $\delta = 10^{-5}$ for FST-FR. We omit some tests for DPSS and ROAST when $N$ is large due to computational limitations.
	}  \label{fig:vary N}
\end{figure*}

\appendix
\section{Proof of Lemma~\ref{lem:singular value decay of FhighB}}
\label{sec:prf singular value decay of FhighB}
\begin{proof}[Proof of Lemma~\ref{lem:singular value decay of FhighB}]
	Note that
	\[
	\mF_N^*\left(\mB_{N,W} - \mF_{N,W}\mF_{N,W}^*\right) = \left[\begin{array}{c}\mF_{N,W}^*\mB_{N,W}- \mF_{N,W}^*\\
	\overline{\mF}_{N,W}^*\mB_{N,W}
	\end{array}\right].
	\]
	By utilizing the result that \[
	\mB_{N,W} = \mF_{N,W} \mF_{N,W}^* + \mL + \mE,
	\]
	where
	\[
	\rank(\mL) \le C_N \log \left( \frac{15}{\eps} \right) \quad \quad \text{and} \quad \quad  \| \mE \| \le \eps,
	\]
	we can rewrite $\overline{\mF}_{N,W}^*\mB_{N,W} = \mL_1 + \mE_1$, where
	\[
	\mL_1: =\overline{\mF}_{N,W}^*\mL  \quad \quad \text{and} \quad\quad  \mE_1: =\overline{\mF}_{N,W}^*\mE .
	\]
	Thus,
	\[
	\rank(\mL_1) \le C_N \log \left( \frac{15}{\eps} \right) \quad \quad \text{and} \quad \quad  \| \mE_1 \| \le \eps.
	\]
	It follows from the Eckart-Young-Mirsky theorem~\cite{eckart1936approximation} that
	\[
	\sigma_{\rank(\mL_1)}\leq \|\mE_1\| \leq \epsilon
	\]
	for any $\eps\in(0,1)$.
	Noting that $\|\overline{\mF}_{N,W}^*\mB_{N,W}\|\leq \|\overline{\mF}_{N,W}^*\| \|\mB_{N,W}\|<1$, we have
	\[\sigma_{\ell} \leq 15e^{-\frac{\ell}{C_N}}.\]
	for all $\ell=0,1,\ldots,N-2\lfloor NW\rfloor -2$. Otherwise, suppose $\sigma_{\ell} > 15e^{-\frac{\ell}{C_N}}$. If $15e^{-\frac{\ell}{C_N}}\geq 1$, then this is in contradiction to the fact that $\sigma_\ell<1$. If $15e^{-\frac{\ell}{C_N}}< 1$, let $\eps = 15e^{-\frac{\ell}{C_N}}$. Then we have a contradiction to the fact that $\sigma_{\rank(\mL_1)}\leq \epsilon$ and $\rank(\mL_1)\leq C_N \log \left( \frac{15}{\eps} \right) = \ell$.
\end{proof}

\section{Proof of Lemma~\ref{lem:guarantee for representing DPSS}}
\label{sec:prf guarantee for representing DPSS}
\begin{proof}[Proof of Lemma~\ref{lem:guarantee for representing DPSS}]
	Fix $K$ to be such that $\lambda_{N,W}^{(K-1)}>\epsilon$. Utilizing $\mB_{N,W} = \mS_{N,W} \mLambda_{N,W} \mS_{N,W}^*$, we have
	\begin{align*}
	&\|\mB_{N,W} - \mQ\mQ^*\mB_{N,W}\|\\ &=\|\mS_{N,W} \mLambda_{N,W} \mS_{N,W}^* - \mQ\mQ^*\mS_{N,W} \mLambda_{N,W} \mS_{N,W}^*\|\\
	&=\|\mLambda_{N,W} - \mS_{N,W}^*\mQ\mQ^*\mS_{N,W} \mLambda_{N,W} \|\\
	&\geq  \|\mLambda_{K} - \mS_K^*\mQ\mQ^*K \mLambda_{K} \|= \|\left( \mId -  \mS_{K}^*\mQ\mQ^*\mS_{K} \right) \mLambda_{K} \|\\
	&\geq  \left\| \mId -  \mS_{K}^*\mQ\mQ^*\mS_{K} \right\|\eps.
	\end{align*}
	On the other hand,
	\begin{align*}
	&\|\mB_{N,W} - \mQ\mQ^*\mB_{N,W}\| \\&=  \left\|\mB_{N,W} - \left[\mF_{N,W} \ \overline{\mF}_{N,W} \mV  \right] \left[\mF_{N,W} \ \overline{\mF}_{N,W} \mV  \right]^\H\mB_{N,W}\right\|\\
	&=  \left\| \overline{\mF}_{N,W}^*\mB_{N,W}- \mV\mV^* \overline{\mF}_{N,W}^* \mB_{N,W} \right\|.
	\end{align*}
	Combining the above two set of equations yields
	\[
	\left\| \mId -  \mS_{K}^*\mQ\mQ^*\mS_{K} \right\| \leq \eta = \frac{  \left\| (\mId- \mV\mV^*) \overline{\mF}_{N,W}^* \mB_{N,W} \right\| }{\epsilon}.
	\]
	Now exploit the relationship between $\mS_K\mS_K^* - \mQ\mQ^*\mS_K \mS_K^*$ and $\mId -  \mS_{K}^*\mQ\mQ^*\mS_{K}$ as follows
	\begin{align*}
	&\left\|\mS_K\mS_K^* - \mQ\mQ^*\mS_K \mS_K^*\right\|^2\\  &= \left\|\left(\mS_K\mS_K^* - \mQ\mQ^*\mS_K \mS_K^*\right)^\T\left( \mS_K\mS_K^* - \mQ\mQ^*\mS_K \mS_K^* \right)\right\| \\
	&= \left\|\mS_K\left(\mId - \mS_K^*\mQ\mQ^*\mS_K\right)\mS_K^\star \right\|\\
	&\leq  \left\|\left(\mId - \mS_K^*\mQ\mQ^*\mS_K\right) \right\|\leq  \eta.
	\end{align*}
	Then, utilizing the inequality $\left\| \mId -  \mS_{K}^*\mQ\mQ^*\mS_{K} \right\|_{\max} \leq \left\| \mId -  \mS_{K}^*\mQ\mQ^*\mS_{K} \right\|$, where $ \left\|\mId -  \mS_{K}^*\mQ\mQ^*\mS_{K} \right\|_{\max}$ is the maximum absolute entry of $\mId -  \mS_{K}^*\mQ\mQ^*\mS_{K}$, we have
	\begin{align*}
	\left|\left(\vs_{N,W}^{(\ell)}\right)^\H\mQ\mQ^*\vs_{N,W}^{(l')}\right| \leq \left\| \mId -  \mS_{K}^*\mQ\mQ^*\mS_{K} \right\|\leq \eta
	\end{align*}
	for all $l\neq l', l,l'=0,1,\ldots,K-1$, and
	\begin{align*}
	&\left\|\vs_{N,W}^{(\ell)} - \mQ \mQ^*\vs_{N,W}^{(\ell)} \right\|_2^2 = 1- \left\|\mQ^*\vs_{N,W}^{(\ell)}\right\|_2^2\\&\leq \left\| \mId -  \mS_{K}^*\mQ\mQ^*\mS_{K} \right\|\leq \eta
	\end{align*}
	for all $l=0,1,\ldots,K-1$.
	
	Let $\vs$ be an arbitrary unit vector in the subspace spanned by $\mS_K$, i.e., $\vs = \sum_{\ell = 0}^{K-1} \alpha_\ell \vs_{N,W}^{(\ell)}$ with $\|\vs\|_2 = \sum_{\ell=0}^{K-1}\alpha_\ell^2 = 1$. We have
	\begin{align*}
	&\left\|\vs - \mQ\mQ^*\vs\right\|_2 = \left\|\sum_{\ell = 0}^{K-1}\alpha_\ell\left(\vs_{N,W}^{(\ell)}-\mQ\mQ^*\vs_{N,W}^{(\ell)}\right)\right\|_2\\
	& \leq \sum_{\ell = 0}^{K-1}\left|\alpha_\ell\right|\left\|\vs_{N,W}^{(\ell)}-\mQ\mQ^*\vs_{N,W}^{(\ell)}\right\|_2\\
	&\leq \sqrt{\eta} \sum_{\ell = 0}^{K-1}\left|\alpha_\ell\right|\leq \sqrt{K\eta} \leq \sqrt{N\eta}
	\end{align*}
	where the last line follows from the inequality between the $\ell_1$-norm and the $\ell_2$-norm: $\|\va\|_1\leq \sqrt{K}\|\va\|_2$ for any $\va\in\R^K$. Thus, we obtain
	\begin{align*}
	\left\|\mQ\mQ^*\vs\right\|_2^2  = 1-  \left\|\vs - \mQ\mQ^*\vs\right\|_2^2 \geq 1-N\eta.
	\end{align*}
	Since this result holds for an arbitrary unit vector $\vs$ in the subspace spanned by $\mS_K$, we finally have
	\[
	\cos(\Theta_{\mS_K,\mQ})  \geq \sqrt{1-N\eta}.
	\]
\end{proof}

\section{Proof of Lemma~\ref{lem:uniform guarantee to pure sinusoid}}
\label{sec:prf uniform guarantee to pure sinusoid}
\begin{proof}[Proof of Lemma~\ref{lem:uniform guarantee to pure sinusoid}]
	Let $\mtx{\Pi}$ be an $N\times N$ diagonal matrix with diagonal entries $j2\pi0, j2\pi, \ldots, j2\pi(N-1)$. The derivative of $\left\|\ve_f - \mU\mU^*\ve_f\right\|_2^2$ in terms of $f$ can be computed as
	\[ \frac{d}{df}\left\|\ve_f - \mU\mU^*\ve_f\right\|_2^2 = 2\real\left(\ve_f^* \left(\mId - \mU\mU^* \right) \mtx{\Pi}\ve_f \right).\]
	We first obtain an upper bound for its derivative
	\begin{align*}
	&\left|\frac{d}{df}\left\|\ve_f - \mU\mU^*\ve_f\right\|_2^2\right| \leq  2\left|\ve_f^* \left(\mId - \mU\mU^* \right) \mtx{\Pi}\ve_f \right|\leq  2 \left|\ve_f^* \mtx{\Pi}\ve_f \right|\|\mId - \mU\mU^* \| \leq  2 \left|\ve_f^* \mtx{\Pi}\ve_f \right| \leq 2\pi N(N-1)\leq 2\pi N^2
	\end{align*}
	for all $f\in[-\frac{1}{2},\frac{1}{2}]$. Since $\left\|\ve_f - \mU\mU^*\ve_f\right\|_2^2$ is nonnegative and its derivative is bounded above, $\left\|\ve_f - \mU\mU^*\ve_f\right\|_2^2$ cannot be too large if $\int_{-W}^W\left\|\ve_f - \mU\mU^*\ve_f\right\|_2^2df$ is very small.
	
	Suppose $\frac{\left\|\ve_f - \mU\mU^*\ve_f\right\|_2^2} {2\pi N^2} \leq 2W$. As illustrated in Figure~\ref{fig:illustration of bounded derivative}, for any $f\in[-W,W]$, we can always find a triangle with area either \[\frac{\left\|\ve_f - \mU\mU^*\ve_f\right\|_2^4} {2\sup_{f\in[-W,W]}\left|\frac{d}{df}\left\|\ve_f - \mU\mU^*\ve_f\right\|_2^2\right|}
	\]
	(the area of the left and right red triangles) or \[
	\frac{\left\|\ve_f - \mU\mU^*\ve_f\right\|_2^4} {\sup_{f\in[-W,W]}\left|\frac{d}{df}\left\|\ve_f - \mU\mU^*\ve_f\right\|_2^2\right|}
	\](the area of the middle red triangle)
	that is smaller than $\int_{-W}^W\left\|\ve_f - \mU\mU^*\ve_f\right\|_2^2df$ (the area under the black curve). This is made more precise as
	\begin{equation}\begin{split}
	&\frac{\left\|\ve_f - \mU\mU^*\ve_f\right\|_2^4}{4\pi N^2}\leq  \frac{\left\|\ve_f - \mU\mU^*\ve_f\right\|_2^4}{2\sup\limits_{f\in[-W,W]}\left|\frac{d}{df}\left\|\ve_f - \mU\mU^*\ve_f\right\|_2^2\right|} \leq \int_{-W}^W\left\|\ve_f - \mU\mU^*\ve_f\right\|_2^2df
	\end{split}\label{eq:bouned derivative integration}\end{equation}
	for all $f\in[-W,W]$.  Thus, we have
	\begin{align*}
	&\frac{\left\|\ve_f - \mU\mU^*\ve_f\right\|_2^2}{\|\ve_f\|_2^2} = \frac{\left\|\ve_f - \mU\mU^*\ve_f\right\|_2^2}{N}\leq 2\sqrt{\pi}\sqrt{ \int_{-W}^W\left\|\ve_f - \mU\mU^*\ve_f\right\|_2^2df}
	\end{align*}
	for all $f\in[-W,W]$.

	On the other hand, suppose $\frac{\left\|\ve_f - \mU\mU^*\ve_f\right\|_2^2} {2\pi N^2} > 2W$.  With a similar argument, as illustrated in Figure~\ref{fig:illustration of bounded derivative2}, for any $f\in[-W,W]$, we can always find a region of area at least $W\left\|\ve_f - \mU\mU^*\ve_f\right\|_2^2 $
	(the area indicated by red dashed lines) that is smaller than $\int_{-W}^W\left\|\ve_f - \mU\mU^*\ve_f\right\|_2^2df$ (the area under the black curve). This is made more precise as
	\begin{equation}
	W\left\|\ve_f - \mU\mU^*\ve_f\right\|_2^2  \leq \int_{-W}^W\left\|\ve_f - \mU\mU^*\ve_f\right\|_2^2df
	\label{eq:bouned derivative integration2}\end{equation}
	for all $f\in[-W,W]$. Thus, we have
	\[
	\frac{\left\|\ve_f - \mU\mU^*\ve_f\right\|_2^2}{\|\ve_f\|_2^2}  \leq \frac{1}{NW} \int_{-W}^W\left\|\ve_f - \mU\mU^*\ve_f\right\|_2^2df
	\]
	for all $f\in[-W,W]$.
	
	\begin{figure}[t]
		\centering
		\includegraphics[width=2in]{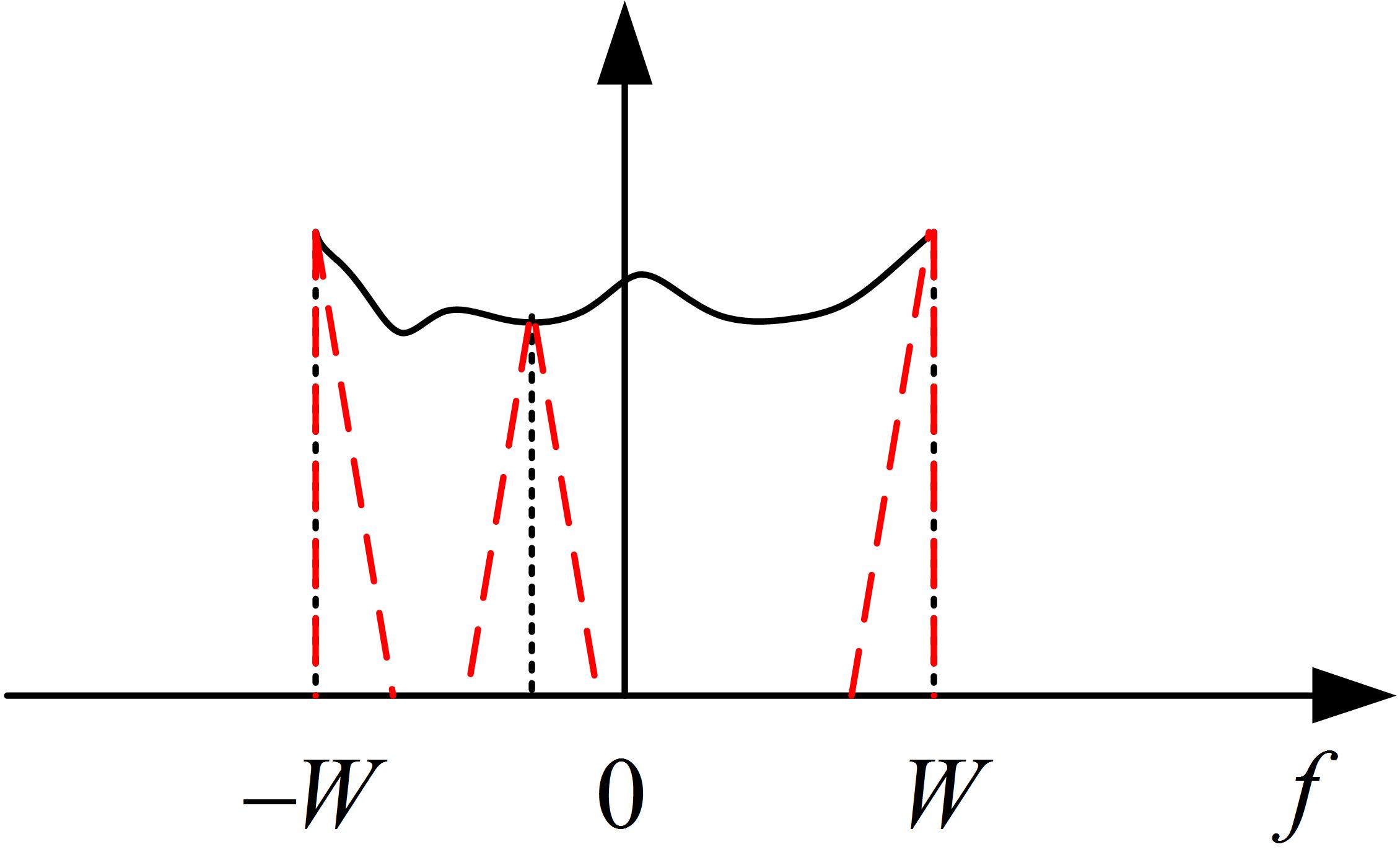}
		\caption{Illustration of \eqref{eq:bouned derivative integration}. The area below the black curve is always larger than or equal to the area of each red triangle.  }\label{fig:illustration of bounded derivative}
	\end{figure}

	\begin{figure}[t]
		\centering
		\includegraphics[width=1.2in]{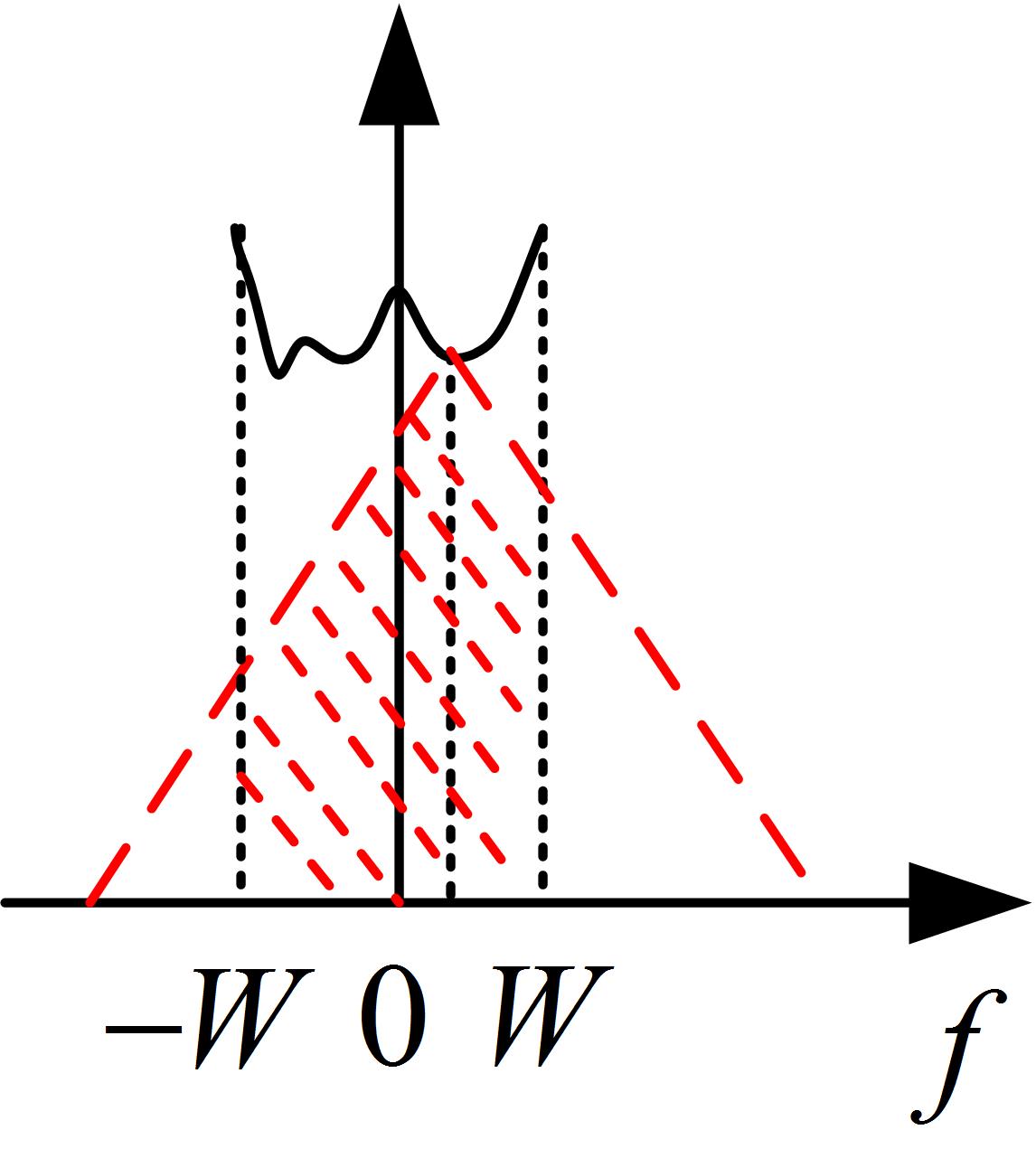}
		\caption{Illustration of \eqref{eq:bouned derivative integration2}. The area below the black curve is always larger than or equal to the area shaded by the red dashed lines.  }\label{fig:illustration of bounded derivative2}
	\end{figure}
\end{proof}

\frenchspacing
\bibliographystyle{ieeetr}
\bibliography{bibfileFAST}

\begin{thebibliography}{10}

\bibitem{zhu2017fast}
Z.~Zhu, S.~Karnik, M.~B. Wakin, M.~A. Davenport, and J.~K. Romberg, ``Fast
  orthogonal approximations of sampled sinusoids and bandlimited signals,'' in
  {\em IEEE Conf. Acous., Speech, Signal Process. (ICASSP)}, pp.~4511--4515,
  2017.

\bibitem{Slepian78DPSS}
D.~Slepian, ``Prolate \uppercase{S}pheroidal \uppercase{W}ave
  \uppercase{F}unctions, \uppercase{F}ourier analysis, and uncertainty.
  \uppercase{V}- \uppercase{T}he discrete case,'' {\em Bell Syst. Tech. J.},
  vol.~57, no.~5, pp.~1371--1430, 1978.

\bibitem{papoulis1975new}
A.~Papoulis, ``A new algorithm in spectral analysis and band-limited
  extrapolation,'' {\em IEEE Trans. Circuits, Systems}, vol.~22, no.~9,
  pp.~735--742, 1975.

\bibitem{hayes1983bandlimited}
M.~Hayes and R.~Schafer, ``On the bandlimited extrapolation of discrete
  signals,'' in {\em Proc. IEEE Int. Conf. Acoust., Speech, and Signal
  Processing (ICASSP)}, vol.~8, pp.~1450--1453, IEEE, 1983.

\bibitem{zemen2005channelEstim}
T.~Zemen and C.~F. Mecklenbr{\"a}uker, ``Time-variant channel estimation using
  \uppercase{d}iscrete \uppercase{p}rolate \uppercase{s}pheroidal
  \uppercase{s}equences,'' {\em IEEE Trans. Signal Process.}, vol.~53, no.~9,
  pp.~3597--3607, 2005.

\bibitem{zemen2007minimum}
T.~Zemen, C.~F. Mecklenbr{\"a}uker, F.~Kaltenberger, and B.~H. Fleury,
  ``Minimum-energy band-limited predictor with dynamic subspace selection for
  time-variant flat-fading channels,'' {\em IEEE Trans. Signal Process.},
  vol.~55, no.~9, pp.~4534--4548, 2007.

\bibitem{amin2008wideband}
M.~G. Amin and F.~Ahmad, ``Wideband synthetic aperture beamforming for
  through-the-wall imaging [lecture notes],'' {\em IEEE Signal Process.
  Magazine}, vol.~25, no.~4, 2008.

\bibitem{AhmadQianAmin2015WallCluterDPSS}
F.~Ahmad, Q.~Jiang, and M.~G. Amin, ``Wall clutter mitigation using
  \uppercase{D}iscrete \uppercase{P}rolate \uppercase{S}pheroidal
  \uppercase{S}equences for sparse reconstruction of indoor stationary
  scenes,'' {\em IEEE Trans. Geosci. Remote Sens.}, vol.~53, no.~3,
  pp.~1549--1557, 2015.

\bibitem{Zhu2015targetDetectDPSS}
Z.~Zhu and M.~B. Wakin, ``Wall clutter mitigation and target detection using
  \uppercase{d}iscrete \uppercase{p}rolate \uppercase{s}pheroidal
  \uppercase{s}equences,'' in {\em 3rd Int. Workshop on Compressed Sensing
  Theory and its Applications to Radar, Sonar and Remote Sensing (CoSeRa)},
  June 2015.

\bibitem{Zhu2016targetDetectDPSS}
Z.~Zhu and M.~B. Wakin, ``On the dimensionality of wall and target return
  subspaces in through-the-wall radar imaging,'' in {\em 4th Int. Workshop on
  Compressed Sensing Theory and its Applications to Radar, Sonar and Remote
  Sensing (CoSeRa)}, September 2016.

\bibitem{Zhu17Asilomar}
Z.~Zhu, D.~Yang, M.~B. Wakin, and G.~Tang, ``A super-resolution algorithm for
  multiband signal identification,'' in {\em 51st Asilomar Conference on
  Signals, Systems and Computers}, (Pacific Grove, California), Oct. 2017.

\bibitem{DavenSSBWB_Wideband}
M.~Davenport, S.~Schnelle, J.~P. Slavinsky, R.~Baraniuk, M.~Wakin, and
  P.~Boufounos, ``A wideband compressive radio receiver,'' in {\em Proc.
  Military Comm. Conf. (MILCOM)}, (San Jose, California), Oct. 2010.

\bibitem{DavenportWakin2012CSDPSS}
M.~A. Davenport and M.~B. Wakin, ``Compressive sensing of analog signals using
  discrete prolate spheroidal sequences,'' {\em Appl. Comput. Harmon. Anal.},
  vol.~33, no.~3, pp.~438 -- 472, 2012.

\bibitem{sejdic2012compressive}
E.~Sejdi{\'c}, A.~Can, L.~F. Chaparro, C.~M. Steele, and T.~Chau, ``Compressive
  sampling of swallowing accelerometry signals using time-frequency
  dictionaries based on modulated \uppercase{d}iscrete \uppercase{p}rolate
  \uppercase{s}pheroidal \uppercase{s}equences,'' {\em EURASIP J. Adv. Signal
  Process.}, vol.~2012, no.~1, pp.~1--14, 2012.

\bibitem{wakin2012nonuniform}
M.~Wakin, S.~Becker, E.~Nakamura, M.~Grant, E.~Sovero, D.~Ching, J.~Yoo,
  J.~Romberg, A.~Emami-Neyestanak, and E.~Candes, ``A nonuniform sampler for
  wideband spectrally-sparse environments,'' {\em IEEE J. Emerg. Sel. Topic
  Circuits Syst}, vol.~2, no.~3, pp.~516--529, 2012.

\bibitem{Karnik2016FAST}
S.~Karnik, Z.~Zhu, M.~B. Wakin, J.~K. Romberg, and M.~A. Davenport, ``The fast
  \uppercase{S}lepian transform,'' {\em {\em to appear in} Appl. Comp. Harm.
  Anal., arXiv preprint arXiv:1611.04950}.

\bibitem{ZhuWakin2015MDPSS}
Z.~Zhu and M.~B. Wakin, ``Approximating sampled sinusoids and multiband signals
  using multiband modulated \uppercase{DPSS} dictionaries,'' {\em J. Fourier
  Anal. Appl.}, vol.~23, pp.~1263--1310, Dec 2017.

\bibitem{businger1969algorithm}
P.~A. Businger and G.~H. Golub, ``Algorithm 358: {S}ingular value decomposition
  of a complex matrix [f1, 4, 5],'' {\em Comm. ACM}, vol.~12, no.~10,
  pp.~564--565, 1969.

\bibitem{halkoRandomizedAlgorithm}
N.~Halko, P.~Martinsson, and J.~A. Tropp, ``Finding structure with randomness:
  {P}robabilistic algorithms for constructing approximate matrix
  decompositions,'' {\em SIAM Rev.}, vol.~53, no.~2, pp.~217--288, 2011.

\bibitem{zemen2017orthogonal}
T.~Zemen, M.~Hofer, D.~Loeschenbrand, and C.~Pacher, ``Orthogonal precoding for
  ultra reliable wireless communication links,'' {\em arXiv preprint
  arXiv:1710.09912}, 2017.

\bibitem{saad2003iterative}
Y.~Saad, {\em Iterative methods for sparse linear systems}.
\newblock SIAM, 2003.

\bibitem{ailon2009fast}
N.~Ailon and B.~Chazelle, ``The fast {J}ohnson--{L}indenstrauss transform and
  approximate nearest neighbors,'' {\em SIAM J. Comput.}, vol.~39, no.~1,
  pp.~302--322, 2009.

\bibitem{hokanson2015}
J.~M. Hokanson, ``Projected nonlinear least squares for exponential fitting,''
  {\em arXiv preprint arXiv:1508.05890}.

\bibitem{duggal2004subspace}
B.~P. Duggal, ``Subspace gaps and range-kernel orthogonality of an elementary
  operator,'' {\em Linear Algebra Appl.}, vol.~383, pp.~93--106, 2004.

\bibitem{bjorck1973numerical}
A.~Bj{\"o}rck and G.~H. Golub, ``Numerical methods for computing angles between
  linear subspaces,'' {\em Math. Comput.}, vol.~27, no.~123, pp.~579--594,
  1973.

\bibitem{mirsky1975trace}
L.~Mirsky, ``A trace inequality of {John von Neumann},'' {\em Monatshefte
  f{\"u}r Mathematik}, vol.~79, no.~4, pp.~303--306, 1975.

\bibitem{eckart1936approximation}
C.~Eckart and G.~Young, ``The approximation of one matrix by another of lower
  rank,'' {\em Psychometrika}, vol.~1, no.~3, pp.~211--218, 1936.

\end{thebibliography}

\end{document}